\documentclass[10pt, conference, letterpaper]{IEEEtran}
\IEEEoverridecommandlockouts

\usepackage{amsmath, amssymb, cite}
\usepackage{amsthm,bm,bbm}
\usepackage{mathtools}
\usepackage[small]{caption}
\usepackage{amsthm,multirow,color,amsfonts}
\usepackage{tabulary}
\usepackage{subfigure}
\usepackage{graphicx}
\usepackage{setspace}
\usepackage{enumerate}
\usepackage[]{algorithm2e}
\usepackage{comment}
\usepackage{balance}

\DeclareMathOperator*{\Hg}{{\bf H}} 

\title{How to Distribute Computation in Networks} 

\author{%
  \IEEEauthorblockN{Derya~Malak}
  \IEEEauthorblockA{Electrical, Computer, and Systems Engineering, RPI\\
                    malakd@rpi.edu}
  \and
  \IEEEauthorblockN{Alejandro~Cohen and Muriel~M\'{e}dard}
  \IEEEauthorblockA{Research Laboratory of Electronics, MIT\\
                    \{cohenale,\,medard\}@mit.edu}
}

\newtheorem{defi}{Definition}

\newtheorem{prop}{Proposition}
\newtheorem{remark}{Remark}

\newtheorem{cor}{Corollary}

\newtheorem{ex}{Example}

\IEEEaftertitletext{\vspace{-2.5\baselineskip}}

\begin{document}
\maketitle

\begin{abstract}
In network function computation is as a means to reduce the required communication flow in terms of number of bits transmitted per source symbol. However, the rate region for the function computation problem in general topologies is an open problem, and has only been considered under certain restrictive assumptions (e.g. tree networks, linear functions, etc.). 
In this paper, we propose a new perspective for distributing computation, and formulate a flow-based delay cost minimization problem that jointly captures the costs of communications and computation. We introduce the notion of entropic surjectivity as a measure to determine how sparse the function is and to understand the limits of computation. 
Exploiting Little's law for stationary systems, we provide a connection between this new notion and the computation processing factor that reflects the proportion of flow that requires communications. This connection gives us an understanding of how much a node (in isolation) should compute to communicate the desired function within the network without putting any assumptions on the topology. Our analysis characterizes the functions only via their entropic surjectivity, and provides insight into how to distribute computation. We numerically test our technique for search, MapReduce, and classification tasks, and infer for each task how sensitive the processing factor to the entropic surjectivity is. 
\end{abstract}

\maketitle

\section{Introduction}
\label{intro}
Challenges in cloud computing include effectively distributing computation to handle the large volume of data with growing computational demand, and the limited resources in the air interface. Furthermore, various tasks such as computation, storage, communications are inseparable. In network computation is required for reasons of dimensioning, scaling and security, where data is geographically dispersed. We need to exploit the sparsity of data within and across sources, as well as the additional sparsity inherent to labeling (function), to provide approximately minimal representations for labeling. 

An equivalent notion to that sparsity is that of data redundancy. Data is redundant in the sense that there exists, a possibly latent and ill understood, sparse representation of it that is parsimonious and minimal, and that allows for data reconstruction, possibly in an approximate manner. Redundancy can occur in a single source of data or across multiple sources.

Providing such sparse representation for the reconstruction of data is the topic of compression, or source coding. The Shannon entropy rate of data provides, for a single source, a measure of the minimal representation, in terms of bits per second, required to represent data. This representation is truly minimal, in the sense that it is achievable with arbitrarily small error or distortion, but arbitrarily good fidelity of reconstruction is provably impossible at lower rates. 

\subsection{Motivation}\label{motivation}

As computation becomes increasingly reliant on numerous, possibly geo-dispersed, sources of data, making use of redundancy across multiple sources without the need for onerous coordination across sources becomes increasingly important. The fact that a minimal representation of data can occur across sources without the need for coordination is the topic of distributed compression. The core result is that of Slepian and Wolf \cite{SlepWolf1973}, who showed that distributed compression without coordination across source can be as efficient, in terms of asymptotic minimality of representation. 

Techniques for achieving compression have traditionally relied on coding techniques. Coding, however, suffers from a considerable cost, as it imputes, beyond sampling and quantization, computation and processing at the source before transmission, then computation and processing at the destination after reception of the transmission. A secondary consideration is that coding techniques, to be efficiently reconstructed at the destination, generally require detailed information about the probabilistic structure of the data being represented. For distributed compression, the difficulty of reconstruction rendered the results in \cite{SlepWolf1973} impractical until the 2000s, when channel coding techniques were adapted.

In the case of learning on data, however, it is not the data itself but rather a labeling of it that we seek. That labeling can be viewed as being a function of the original data. The reconstruction of data is in effect a degenerate case where the function is identity. Labeling is generally a highly surjective function and thus induces sparsity, or redundancy, in its output values beyond the sparsity that may be present in the data. 
 
The use of the redundancy in both functions and data to provide sparse representations of functions outputs is the topic of the rather nascent field of functional compression. A centralized communication scheme requires all data to be transmitted to some central unit in order to perform certain computations. However, in many cases such computations can be performed in a distributed manner at different nodes in the network avoiding transmission of unnecessary information in the network. Hence, intermediate computations can significantly reduce the resource usage, and this can help improve the trade-off between communications and computation.

\begin{figure*}[t!]
\centering
\includegraphics[width=\textwidth]{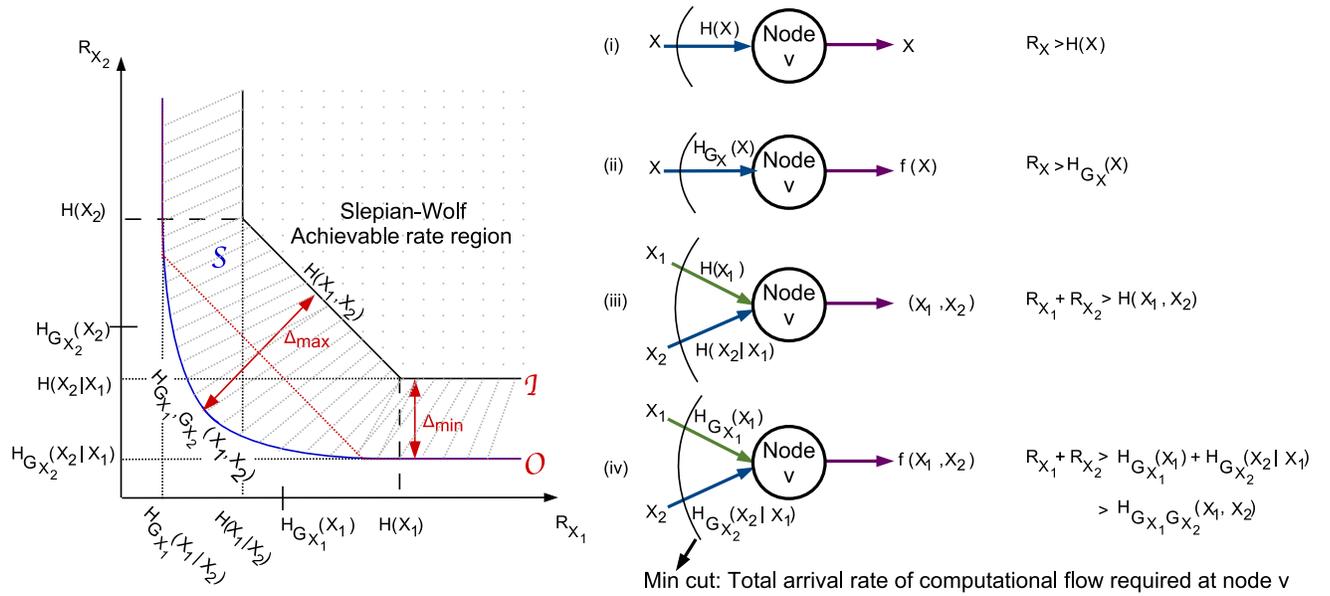}
\caption{(Left) Example rate region for the zero distortion distributed functional compression problem \cite{DosShaMedEff2010}. $\mathcal{S}$ denotes the shaded region between the joint entropy $H(X_1,X_2)$ curve (inner bound $\mathcal{I}$) and the joint graph entropy curve of $H_{G_{X_1},G_{X_2}}(X_1,X_2)$ (outer bound $\mathcal{O}$). Note that any point above $\mathcal{I}$ is the Slepian-Wolf achievable rate region, and $\mathcal{O}$ is characterized by how surjective the graph entropy is. (Right) Example scenarios with achievable rates: rate region for (i) source compression, (ii) functional compression, (iii) distributed source compression with two transmitters and a receiver, and (iv) distributed functional compression with two transmitters and a receiver. Note that in (iv) the main benefit of joint graph entropy $H_{G_{X_1},G_{X_2}}(X_1,X_2)$ is that it is less than the sum of the marginal graph entropy of source $X_1$, i.e. $H_{G_{X_1}}(X_1)$, and the conditional graph entropy of source $X_2$ given $X_1$, i.e. $H_{G_{X_2}}(X_2|X_1)$. Joint graph entropy provides a better rate region than the joint entropy since it does not satisfy the chain rule. Hence, we expect to have $\Delta_{\max}>\Delta_{\min}$ in the left figure.}
\label{SWandGraphEntropy_surjectivity_details}
\end{figure*} 	

\subsection{Technical Background}
\label{compresstocompute}
In this section, we introduce some concepts from information theory which characterize the minimum communication (in terms of rate) necessary to reliably evaluate a function. In particular, this problem which is referred to as distributed functional compression, has been studied under various forms since the pioneering work of Slepian and Wolf \cite{SlepWolf1973}. 

An object of interest in the study of these fundamental limits is the characteristic graph, and in particular its coloring.
In the characteristic graph, each vertex represent a possible different sample value, and two vertices are connected if they should be distinguished. More precisely, for a collection of random variables $X_1,\ldots,X_n$ assumed to take values in the same alphabet $\mathcal{X}$, and a function $g: \mathcal{X} \to \mathcal{Y}$, we draw an edge between vertices $u$ and $v \in \mathcal{X}$, if $g(u, x_2,\ldots,x_n) \neq g(v, x_2,\ldots, x_n)$ for any $x_2,\ldots,x_n$ whose joint instance has non-zero measure. We illustrate the characteristic graph and its relevance in compression through the following example.

\paragraph{Slepian-Wolf Coding (or Compression)}
We start by reviewing the natural scenario where the function $f(X_1,\ldots,X_n)$ is the identity function, i.e., the case of distributed lossless compression. For sake of presentation, we focus on the case of two random variables $X_1$ and $X_2$, which are jointly distributed according to $P_{X_1,X_2}$.
Source random variable $X_1$ can be asymptotically compressed up to the rate $H(X_1|X_2)$ when $X_2$ is available at the receiver \cite{SlepWolf1973}. Given two statistically dependent i.i.d. finite alphabet sequences $X_1$ and $X_2$, the Slepian-Wolf theorem gives a theoretical bound for the lossless coding rate for distributed coding of the two sources as shown below \cite{SlepWolf1973}:
\begin{align}
\label{rateregionSW}
R_{X_1} \geq H(X_1|X_2),\quad
R_{X_2} \geq H(X_2|X_1)\nonumber\\
R_{X_1}+R_{X_2} \geq H(X_1,X_2).
\end{align}
We denote the rate region in (\ref{rateregionSW}) by $\mathcal{R}(X_1,X_2)$. The Slepian-Wolf theorem states that in order to recover a joint source $(X_1,X_2)$ at a receiver, it is both necessary and sufficient to encode separately sources  $X_1$ and $X_2$  at rates $(R_{X_1}, R_{X_2})$ where (\ref{rateregionSW}) is satisfied \cite{DosShaMedEff2010}. Note that the encoding is done in a truly distributed way, i.e. no communication or coordination is necessary between the encoders. Distributed coding can achieve arbitrarily small error probability for long sequences.

One of the challenge in function computation is the function $f$ on the data $X$ itself. Whether or not having correlations among the source random variables (or data) $X$, due to the mapping from the sources to the destinations the codebook design becomes very challenging. Since the rate region of the distributed function computation problem depends on the function, designing achievable schemes for the optimal rate region for function computation (or  compression) (for general functions, with/without correlations) remains an open problem. We aim to develop a tractable approach for computing general functions using the tools discussed next.

\paragraph{Graph Entropy for Characterizing the Rate Bounds}
Given a graph $G_{X_1} = (V_{X_1} , E_{X_1} )$ and a distribution on its vertices $V_{X_1}$, the graph entropy is expressed as
\begin{align}
H_{G_{X_1}}(X_1)= \min\limits_{X_1\in W_1\in\Gamma(G_{X_1})} I(X_1; W_1), 
\end{align}
where $\Gamma(G_{X_1})$ is the set of all maximal independent sets of $G_{X_1}$. The notation $X_1 \in W_1 \in \Gamma(G_{X_1})$ means that we are minimizing over all distributions $p(w_1, x_1)$ such that $p(w_1,x_1) > 0$ implies $x_1 \in w_1$, where $w_1$ is a maximal independent set of the graph $G_{x_1}$.

In \cite[Theorem 41]{FeiMed2014}, authors have determined the rate region for a distributed functional compression problem with two transmitters and a receiver. This rate region is given by 
\begin{align}
\label{rateregiongraph}
R_{11} \geq H_{G_{X_1}}(X_1|X_2),\quad
R_{12} \geq H_{G_{X_2}}(X_2|X_1)\nonumber\\
R_{11}+R_{12} \geq H_{G_{X_1},G_{X_2}}(X_1,X_2),
\end{align}
where $G_X$ is the characteristic graph of $f$ on the data $X$, and $H_{G_{X_1},G_{X_2}}(X_1,X_2)$ is the joint graph entropy of the sources.

To summarize, the role of in network function computation is to reduce the amount of rate needed to be able to recover a function on data, and the amount of reduction is observed as
\begin{equation}
H(X) \to H_G(\cdot) \to H_{G_X}(X).
\end{equation}

An achievable scheme for the above functional compression problem has been provided in \cite{FeiMed2014}. In the scheme, the sources compute colorings of high probability subgraphs of their $G_X$ and perform source coding on these colorings and send them. Intermediate nodes compute the colorings for their parents', and by using a look-up table (to compute their functions), they find corresponding source values of received colorings.

In Figure \ref{SWandGraphEntropy_surjectivity_details}, we illustrate the Slepian-Wolf compression rate region $\mathcal{I}$ in (\ref{rateregionSW}) versus the outer bound $\mathcal{O}$ (convex) determined by the joint graph entropy of variables $X_1$ and $X_2$, as given in (\ref{rateregiongraph}). In the graph, the region between two bounds, denoted by $\mathcal{S}$, determines the limits of the functional compression. We denote the depth of this region by $\delta$ that satisfies $\delta\in[\Delta_{\min},\Delta_{\max}]$. This region indicates that there could be potentially a lot of benefit in exploiting the compressibility of the function to reduce communication. The convexity of $\mathcal{O}$ of $\mathcal{S}$ can be used to exploit the tradeoff between communications and computation, which is mainly determined by the network, data and correlations, and functions. A notion of compressibility is the deficiency metric introduced in  \cite{FuaFenWanCar2018}.

\begin{defi}{\bf Deficiency \cite{PanSakSteWan2011}.}
Let $G_1$ and $G_2$ be finite Abelian groups of the same cardinality $n$ and $f: G_1 \to G_2$. Let $G_1^{*} =G_1\backslash\{0\}$ and $G_2^{*} =G_2\backslash\{0\}$. For any $a\in G_1^{*}$ and $b\in G_2^{*}$, we denote
$\Delta_{f,a} (x) = f(x+a)-f(x)$ and $\lambda_{a,b}(f) = \#\Delta_{f,a}^{-1}(b)$. Let 
$\alpha_i(f)=\#\{(a,b)\in G_1^{*}\times G_2\vert \lambda_{a,b}(f)=i\}$ for $0\leq i\leq n$. We call $\alpha_0(f)$ the deficiency of $f$. Hence $\alpha_0(f)$ measures the number of pairs $(a,b)$ such that $\Delta_{f,a}(x) = b$ has no solutions. This is a measure of the surjectivity of $\Delta_{f,a}$; the lower the deficiency the closer the $\Delta_{f,a}$ are to surjective.
\end{defi}

Although Figure \ref{SWandGraphEntropy_surjectivity_details} gives insights on the limits of compression, it is not clear which point in the outer bound $\mathcal{O}$ provides the best solution from a joint optimization of communication and computation. In particular, as highlighted in Sect. \ref{compresstocompute}, constructing optimal compression codes imposes a significant computational burden on the encoders and decoders, since the achievable schemes are based on NP-hard concepts. If the cost of computation were insignificant, it would be optimal to operate at max $\Delta$. However, when the computation cost is not negligible, there will be a strain between the costs of communication and computation. To capture this balance, we propose to follow a different approach, as detailed in Sect. \ref{networkmodel}.

	\subsection{Contributions}	
	\label{contributions}
	The function computation task in networks is very challenging, and to the best of our knowledge, is unknown except for special cases as outlined in Sect. \ref{relatedwork}. In this paper, we provide a fresh look at this problem from a networking perspective. 
	
	Our contributions are as follows. We provide a cost model for a general network topology for performance characterization of distributed function computation by jointly considering the computation and communications aspects. We introduce entropic surjectivity as a measure to determine how sparse a function is. We devise a flow-based delay cost minimization technique that incorporates the costs of communications and computation. While we assume that the communications cost is convex in flow, we use general cost functions for computation. The enabler of our approach is the connection between Little's law for stationary systems and proportion of flow that requires communications (i.e. computation processing factor) that is determined by the entropic surjectivity of functions.
				  
Our goal is to employ/devise distributed (function) compression techniques in general network topologies (stationary and Jackson type networks where the approach allows for the treatment of individual nodes in isolation, independent of the network topology. Therefore, we do not have to restrict ourselves to cascading operations as in \cite{FeiMed2010allerton} due to the restriction of topology to linear operations.) as a simple means of exploiting function's entropic surjectivity (a notion of sparsity inherent to labeling), by employing the concepts of graph entropy, in order to provide approximately minimal representations for labeling. Labels can be viewed as colors on characteristic graph of the function on the data, where in our case the labeling is the function, is central to functional compression\footnote{The entropy rate of the coloring of the function's power conflict graph upon vectors of data characterizes the minimal representation needed to reconstruct with fidelity the desired function of the data \cite{Korner1973}.  The degenerate case of the identity function corresponds to having a complete characteristic graph.}. Our main insight is that, the main characteristics required for operating the distributed computation scheme are those associated with the entropic surjectivity of the functions. 

The advantages of the proposed approach is as follows. It does not put any assumptions on the network topology and characterizes the functions only via their entropic surjectivity, and provides insight into how to distribute computation/compression depending on the entropic surjectivity of the computation task, how to distribute computation, and how to use the available resources among different computation tasks, and how it compares with the centralized solution. Our results imply that most of the available resources will go to the computation of low complexity functions and fewer resources will be allocated to the processing of high complexity functions.	
	
The organization for the rest of the paper is as follows. In Sect. \ref{relatedwork}, we review the related work. In Sect. \ref{networkmodel}, we detail how to model computation, and derive some lower bounds on the rate of generated flows (i.e. processing factors) of the nodes by linking the computation problem to Little's law. In Sect. \ref{performance}, we present numerical results and discuss possible directions.

\section{Related Work}
\label{relatedwork}
Compressed sensing and information theoretic limits of representation provide a solid basis for function computation in distributed environments. Problem of distributed compression has been considered from different perspectives. For source compression, distributed source coding using syndromes (DISCUS) have been proposed \cite{PradRam2003}, and source-splitting techniques have been discussed \cite{ColLeeMedEff2006}. For data compression, there exist some information theoretic limits, such as side information problem \cite{WynZiv1976}, Slepian-Wolf coding or compression for depth-one trees \cite{SlepWolf1973}, which can be generalized to trees, and  general networks via multicast and random linear network coding \cite{HoMedKoeKarEffShiLeo2006}. 

In functional compression, a function of sources is sought at destination. Korner introduced graph entropy \cite{Korner1973}, which was used in characterizing rate bounds in various functional compression setups \cite{AlonOrlit1996}. For a general function and a configuration where one source is local and another collocated with the destination, Orlitsky and Roche provided a single-letter characterization of the rate-region in \cite{OrlRoc2001}. In \cite{DosShaMedEff2010} and \cite{FeiMed2014}  authors investigated graph coloring approaches for tree networks. In \cite{FES04} authors computed a rate-distortion region for  functional compression with side information. Another class of work considered the in network computation problem for some specific functions. In \cite{KowKum2010} authors investigated the cut-set bounds for the computation 
of symmetric Boolean functions in tree networks. The asymptotic analysis of the rate in broadcast networks has been conducted in \cite{Gal88}, and in random geometric graphs in \cite{KM08}. Function computation has been studied using multi-commodity flow techniques in \cite{ShaDeyMan2013}. There do not exist, however, tractable approaches to perform functional compression in ways that approximate the information theoretic limits. Thus, unlike the case for compression, where coding techniques exist and compressed sensing acts in effect as an alternative for coding, for purposes of simplicity and robustness, there are currently no family of coding techniques for functional compression.

Computing capacity of a network code is the maximum number of times the target function can be computed for one use of the network \cite{HuanTanYangGua2018}. This capacity for special cases such as trees, identity function \cite{LiYeuCai2003}, linear network codes to achieve the multicast capacity have been studied \cite{LiYeuCai2003}, \cite{KoeMed2003}. For scalar linear functions, the computing capacity can be fully characterized by min cut \cite{KoeEffHoMed2004}. For vector linear functions over a finite field, necessary and sufficient conditions have been obtained so that linear network codes are sufficient to calculate the function \cite{AppusFran2014}. For general functions and network topologies, upper bounds on the computing capacity based on cut sets have been studied \cite{KowKum2010}, \cite{KowKum2012}. In \cite{HuanTanYangGua2018}, authors generalize the equivalence relation for the computing capacity. However, in these papers, characterizations based on the equivalence relation associated with the target function is only valid for special network topologies, e.g., the multi-edge tree. For more general networks, this equivalence relation is not sufficient to explore the general function computation problems. 

Coding for computation have been widely studied in the context of multi-stage computations \cite{LiAliYuAves2018} which generally focus on linear reduce functions (since many reduce functions of interest are linear); heterogeneous networks and asymmetric computations \cite{KiaWanAves2017}, and compressed coded computing \cite{LiAliYuAves2018}, which focused on computations of single-stage functions in networks. Coded computing aims to tradeoff the communication (bottlenecks) by injecting computations. While fully distributed algorithms might cause a high communication load, fully centralized systems can suffer from high computation load. With distributed computing at intermediate nodes by exploiting multicast coding opportunities, the communication load can be significantly reduced, can be made inversely proportional to the computation load \cite{LiAliYuAves2018}, \cite{kamran2019deco}. The rate-memory tradeoff for function computation has been studied in \cite{YuAliAves2018}. Different coding schemes to improve the recovery threshold include Lagrange coded computing \cite{YuRavSoAve2018}, and polynomial codes for distributed matrix multiplication \cite{YuAliAve2017}.   

In functional compression, functions themselves can also be exploited. There exist functions with special structures, such as sparsity promoting functions \cite{SheSutTri2018}, symmetric functions, type sensitive and threshold functions \cite{GK05}. One can also exploit a function's surjectivity. There are different notions on how to measure surjectivity, such as deficiency \cite{FuaFenWanCar2018}, ambiguity \cite{PanSakSteWan2011}, and equivalence relationships among function families \cite{gorodilova2019differential}.

\section{Modeling Computation in Networks}
\label{networkmodel}
In this section, we want to answer the following questions: How to handle large, distributed data? What is the rate region of the distributed functional compression problem for a general function? Where to place computation and memory? When to do computations? How to model computation in networks?

As a first step to ease this problem, we will provide a utility-based approach for general cost functions. As special cases, we continue with simple example of point search ($O(\log N)$), then MapReduce $(O(N))$, then the binary classification model ($O(\exp(N))$). Our main contribution is to provide the link between the function computation problem and Little's law. 
	 
	We consider a general stationary network topology. While sources can be correlated, and computations are allowed at intermediate nodes, we compute some deterministic functions. Our goal is to effectively distribute computation. Intermediate nodes need to decide whether to compute or relay. At each node, computation is followed by computation (causality) while satisfying stability conditions. We consider a decentralized solution. This yields a threshold of flow (i.e. processing factor) to be able to perform computation. We also consider a centralized solution which can be obtained by solving an optimization problem by using appropriate cost functions.		
	
	We use the following notation. The set of source random variables is denoted by $X=X_1^N=X_1,\hdots X_N$. Arrival rate of type $c$ flow at node $v$ is $\lambda_v^c$. Service rate of type $c$ flow at node $v$ is given by $\mu_v^c$. Average number of packets at node $v$ due to the processing of type $c$ function is $M_v^c$. 
	Function of type $c$ is denoted by $f_c(X_1^N)$. In this section and in the remaining of the paper, we drop the subscript in graph entropy  $H_{G_X}(X)$, and instead use the boldface notation $\Hg(f(X))$ to show the dependency of the graph entropy on the function $f$ on the data $X$. Hence, the (graph) entropy of function $f_c(X_1^N)$ is $\Hg(f_c(X_1^N))$. Time complexity of generating/processing a flow of type $c$ at node $v$ is $d_f(M_v^c)$. The generation rate of the flow, i.e. the processing factor, of type $c$ at node $v$ is $\gamma_f(\lambda_v^c)$.

\subsection{Computing with Little's Law}
	
In this section, we connect the computation problem to Little's law. Little's law states that the long-term average number $L$ of packets in a stationary system is equal to the long-term average effective arrival rate $\lambda$ multiplied by the average time $W$ that a packet spends in the system. More formally, it can be expressed as $L=\lambda W$. The result applies to any system that is stable and non-preemptive, and the relationship does not depend on the distribution of the arrival process, the service distribution, and the service order \cite{Klein1975}.
		
In our setting, the average time a packet spends in the system is given by the addition of the total time required by computation followed by the total time required by communications. We formulate a utility-based optimization problem by decoupling the costs of communications and computation:
	\begin{equation}
	\label{costoptimization}
	\begin{aligned}
	{\rm MinCost}: \,\,& \underset{\rho}{\min}
	& & C = \sum\limits_{v\in V} \sum\limits_{c\in\mathcal{C}} W_{v}^c\\	
	& \hspace{0.3cm}\text{s.t.}
	& & \rho_v^c< 1, \quad \forall c\in \mathcal{C},\,\, v\in V,
	\end{aligned}
	\end{equation}		
	where $W_v^c=C_{v,comp}^c+C_{v,comm}^c$ captures the total delay, and 
	$C_{v,comp}^c$ and $C_{v,comm}^c$
	are positive delay cost functions that are non-decreasing in flow. The delays of computation and communications for processing functions of type $c\in\mathcal{C}$ are 
	\begin{align}
	\label{delaycostperflow}
	C_{v,comp}^c=\frac{1}{\lambda_v^c}d_f(M_v^c),\quad
	C_{v,comm}^c=\frac{1}{\mu_v^c-\gamma_f(\lambda_v^c)},
	\end{align}
	where $d_f$ models the time complexity of computation, i.e. the total time needed to process all the incoming packets and generate the desired function outcomes. The term $\gamma_f(\lambda_v^c)$ characterizes the amount of computation flow rate generated by node $v$ for function of type $c$. Hence, the second term on the right hand side captures the waiting time, i.e. the queueing and service time of a packet. Hence, by Little's law, we expect that the long-term average number $L_v^c$ of packets in node $v$ for function of type $c\in\mathcal{C}$ satisfies the following relation
	\begin{align}
	\label{LittleComputation}
	L_v^c=\gamma_f(\lambda_v^c)W_{v}^c,
	\end{align}
	where we aim to infer the value of $\gamma_f(\lambda_v^c)$ using Little's law.

	The connection between $L_v^c$ and $M_v^c$ can be given as 
	\begin{align}
	M_v^c=L_v^c\left(1-{\gamma_f(\lambda_v^c)}/{\lambda_v^c}\right).
	\end{align}
	
		For simplicity of notation, let $\rho_v^c={\lambda_v^c}/{\mu_v^c}\in[0,1)$, and $\rho=[\rho_v^c]_{c\in \mathcal{C},\,v\in V}$, $\lambda=[\lambda_v^c]_{c\in \mathcal{C},\,v\in V}$, and $\mu=[\mu_v^c]_{c\in \mathcal{C},\,v\in V}$. 
	
	The following gives a characterization of $L_v^c$ by simple lower and upper bounding techniques.
	\begin{prop}\label{FlowResult} {\bf Flow bounds.} The long-term average number of packets in $v$ for type $c$ flow satisfies
	\begin{align}
	\label{FlowBounds}
	\hspace{-0.31cm}\frac{\Hg(f_c(X_1^N))}{2}\!+\!1\!-\sqrt{\Big(\frac{\Hg(f_c(X_1^N))}{2}\Big)^2+1} \leq L_v^c\leq M_v^c.
	\end{align}
	\end{prop}
	
	Prop. \ref{FlowResult} yields a better inner bound than that of Slepian and Wolf \cite{SlepWolf1973} because the LHS of (\ref{FlowBounds}) is always less than or equal to $\Hg(f_c(X_1^N))$. Its proof is provided in Appendix. 

	We next provide a result required for stability.
	\begin{prop}
	For stability, we require that $d_f(M_v^c)>M_v^c$.
	\end{prop}
	
	\begin{proof}
	Assume that $C_{v,comp}^c<C_{v,comm}^c$. We then have
	\begin{align}
	C_{v,comp}^c=\frac{1}{\lambda_v^c}d_f(M_v^c)\!<\!C_{v,comm}^c=\frac{1}{\mu_v^c-\gamma_f(\lambda_v^c)}\!<\!\frac{1}{\mu_v^c-\lambda_v^c},\nonumber
	\end{align}
	where the rightmost term is the total cost in the case of no computation. Hence, if $d_f(M_v^c)<M_v^c$, then we have $C_{v,comp}^c<C_{v,comm}^c$. In other words, the number of packets waiting for communications is higher than the number of packets waiting for computation. In this case, packets will accumulate while waiting for communication service, which will violate the stability condition. Hence, delay  of  computation should  be  higher, i.e. $d_f(M_v^c)>M_v^c$ is required.
	\end{proof}

	\begin{prop}{\bf Rate of generated flow.} \label{rateofcomputationflow1}
	The processing factor of node $v$ for type $c$ flow is given by
	\begin{align}
	\gamma_f(\lambda_v^c)\geq a_v^c\pm \sqrt{(a_v^c)^2-\lambda_v^c\mu_v^c}
	\end{align}
	where $2a_v^c= \lambda_v^c+\mu_v^c+\lambda_v^c/d_f(M_v^c)$. Hence,
	\begin{align}
	L_v^c&\geq (a_v^c- \sqrt{(a_v^c)^2-\lambda_v^c\mu_v^c})\nonumber\\ &\cdot\Big[\frac{1}{\lambda_v^c}d_f(M_v^c)+\frac{1}{\mu_v^c-(a_v^c- \sqrt{(a_v^c)^2-\lambda_v^c\mu_v^c})}\Big].\nonumber
	\end{align}
	\end{prop}

	\begin{proof}
	Via computation, we aim to achieve $L_v^c=O(\sqrt{d_f(M_v^c)})$, $M_v^c=\frac{\lambda_v^c}{\mu_v^c(1-\rho_v^c)}$. From 
	Little's law (\ref{LittleComputation}) we have
	\begin{align}
	\label{delaycomputation}
	 d_f(M_v^c) \leq L_v^c= \gamma_f(\lambda_v^c) \Big[\frac{d_f(M_v^c)}{\lambda_v^c}+\frac{1}{\mu_v^c-\gamma_f(\lambda_v^c)}\Big].
	\end{align}
	Simplifying the above relation, we get:
	\begin{align}
	\hspace{-0.3cm}\gamma_f(\lambda_v^c)^2-(\lambda_v^c+\mu_v^c+\lambda_v^c/d_f(M_v^c) )\gamma_f(\lambda_v^c)+\lambda_v^c\mu_v^c\geq 0.\nonumber
	\end{align}
	Simplifying above we get the desired result.
	\end{proof}

	\begin{remark}
	From Prop. \ref{rateofcomputationflow1}, observe that as the time complexity $d_f(M_v^c)$ of computation increases, $a_v^c$ decreases and the generated flow amount decreases. Ignoring this principle, if the processed flow rate were increased with the time complexity of the function, then the cost for both computation and communications would increase together. However, the processing factor 
	can decrease with the time complexity, and the output rate may not be compressed below $\Hg(f_c(X))$. Hence, the value of the processed flow should satisfy 
	\begin{align}
	\Hg(X_1^N)\geq L_v^c\geq \Hg(f_c(X)).\nonumber
	\end{align}
	\end{remark}

\subsection{Entropic Surjectivity}
	
	In our context, entropic surjectivity is a measure of how well a network can compress a function that the destination wants to compute. Since non-surjective functions have low entropy, a function with high entropy yields a high entropic surjectivity. Hence, for surjective functions $\Hg(f_c(X))/H(X)\approx 1$ and it is not possible to do much further compression. 
	
	\begin{defi}{\bf Entropic surjectivity, $\Gamma_c(f)$.} 
Entropic surjectivity of a function is how well the function $f_c: X\to Y$ can be compressed with respect to the compression rate of its domain $X$. We denote the entropic surjectivity of function $f_c$ with respect to source symbols $X$ by
\begin{align}
\Gamma_c(f)=\Gamma_c={\Hg(f_c(X))}\big/{H(X)},
\end{align}
where we emphasize that $\Gamma_c$ is a function of the function $f_c$. 
\end{defi}	
Note that $\Gamma_c$ is maximized when the function $f$ with domain $X$ and codomain $Y$ is surjective, i.e. for every $y\in Y$ there exists at least one $x\in X$ with $f(x)=y$. Note also that $\Gamma_c$ is lower bounded by zero which is when the function maps all elements of $X$ to the same element of $Y$. Therefore, $\Gamma_c$ can be used a measure of how surjective the function $f_c$ is.
	
	Consider a function associated with class $c$, i.e. $f_c: X\to f_c(X)$. Total incoming flow rate needed  (bits/source symbol required) can be approximated as $H(X)$. However, to be able to compute $f_c(X)$, we need to transmit at least $\Hg(f_c(X))$ bits/source symbols. In this case, the proportion of flow that requires communications (which is the same as the proportion of flow that is generated as a result of computation task):
	\begin{align}
	\label{flow_entropy_relation}
	\Gamma_c\approx \gamma_f(\lambda_v^c)\big/\lambda_v^c.
	\end{align}	
	Our objective is to bound $\gamma_f(\lambda_v^c)$ using the connection between Little's law that connects the number of packets $L_v^c$ with the entropic surjectivity of the function $\Gamma_c$. Given the surjectivity, maximum amount of reduction in communications flow that can be handled is $H(X)-\Hg(f_c(X))=H(X)(1-\Gamma_c)$.				

\begin{defi}{\bf Set of computational flows.}
We denote the set of computational flows by $\mathcal{C} = \{c = (f, X) \}$ where $f\in\mathcal{F}$ which represents the class of functions and $X\in\mathcal{X}$ defined on the probability space $(\mathcal{X},\mathcal{P})$ where $\mathcal{X}$ is the set of symbols and $\mathcal{P}$ is the data (or source) distribution.
\end{defi}

\begin{defi}{\bf Processing (surjection) factor.} Processing factor of a node $v$ is the computational flow rate generated by $v$ as a result of computing $f_c$. It is given by $\gamma_f=\gamma_f(\lambda_v^c)\geq \lambda_v^c \Gamma_c$.
\end{defi}

In Figure \ref{flowpernode}, we illustrate different components of computational aspects at a typical node $v\in V$ of the network. Since the network is of Jackson type, we can consider a node in isolation. Note that the min-cut that denotes the total arrival rate of computational flow $c$ is given by $\lambda_v^c$. This via (\ref{relation_lambdav_lambdavout}) captures the rate of original arrivals which is $\beta_v^c$, and the arrivals routed from any other node $v'\in V$ in the network. If there is no $v'\in V$ such that $p_{v',v}^{rou}(c)>0$, then $\lambda_v^c=\beta_v^c$. The cut $\gamma_f(\lambda_v^c)$ denotes the total generated rate (or processing factor) of computational flow $c$ at node $v$. The processed flow can be routed to any $v''\in V$ in the network if $p_{v,v''}^{rou}(c)>0$. If there is no such node, then $\gamma_f(\lambda_v^c)$ departs the system.

\begin{prop}\label{LoadThreshold}{\bf Load threshold for distributed function computation.}
A node $v\in V$ can do computation of a class $c\in\mathcal{C}$ function if the following condition is satisfied:
\begin{align}
\frac{(\rho_v^c)^2}{1-\rho_v^c}\!>\! f_{\rm comp}^d(M_v^c)\frac{1-\rho_v^c{\Hg(f_c(X_1^N))}/{H(X_1^N)}}{1-{\Hg(f_c(X_1^N))}/{H(X_1^N)}}\!>\! d_f(M_v^c).\nonumber
\end{align}

A more relaxed threshold on computation is given by
\begin{align}
\label{relaxed_condition_computation}
\rho_v^c >\rho_{th}=\sqrt{\Big(\frac{d_f(M_v^c)}{2}\Big)^2+d_f(M_v^c)}-\frac{d_f(M_v^c)}{2}.
\end{align}
\end{prop}	

\begin{proof}
The threshold $\rho_{th}$ is obtained by comparing the total delay in (\ref{delaycostperflow}) with and without computation, on a per node basis. 
If the delay 
caused only by communication is higher than the total delay caused by computation followed by communication $W_v^c$ in (\ref{delaycostperflow}), i.e. the following condition is satisfied at node $v$, the node decides that computation is required:
\begin{align}
\label{computation_cost_lower}
\frac{1}{\mu_v^c(1-\rho_v^c)} > \frac{1}{\lambda_v^c}d_f(M_v^c)+\frac{1}{\mu_v^c\left(1-\rho_v^c\Gamma_c\right)},
\end{align}
where we used $\gamma_f(\lambda_v^c)=\lambda_v^c\Gamma_c$ on RHS. From (\ref{computation_cost_lower}), we have
\begin{align}
\label{gamma_vs_d_relation}
\frac{(\rho_v^c)^2}{1-\rho_v^c}> d_f(M_v^c)\frac{1-\rho_v^c\Gamma_c}{1-\Gamma_c}>d_f(M_v^c).
\end{align}
Using $\frac{(\rho_v^c)^2}{1-\rho_v^c}>d_f(M_v^c)$, we get the relaxed condition for computation in (\ref{relaxed_condition_computation}). From (\ref{gamma_vs_d_relation}) observe $\rho_{th}\to 1$ as $M_v^c\to\infty$.
\end{proof}

\begin{figure}[t!]
\centering
\includegraphics[width=\columnwidth]{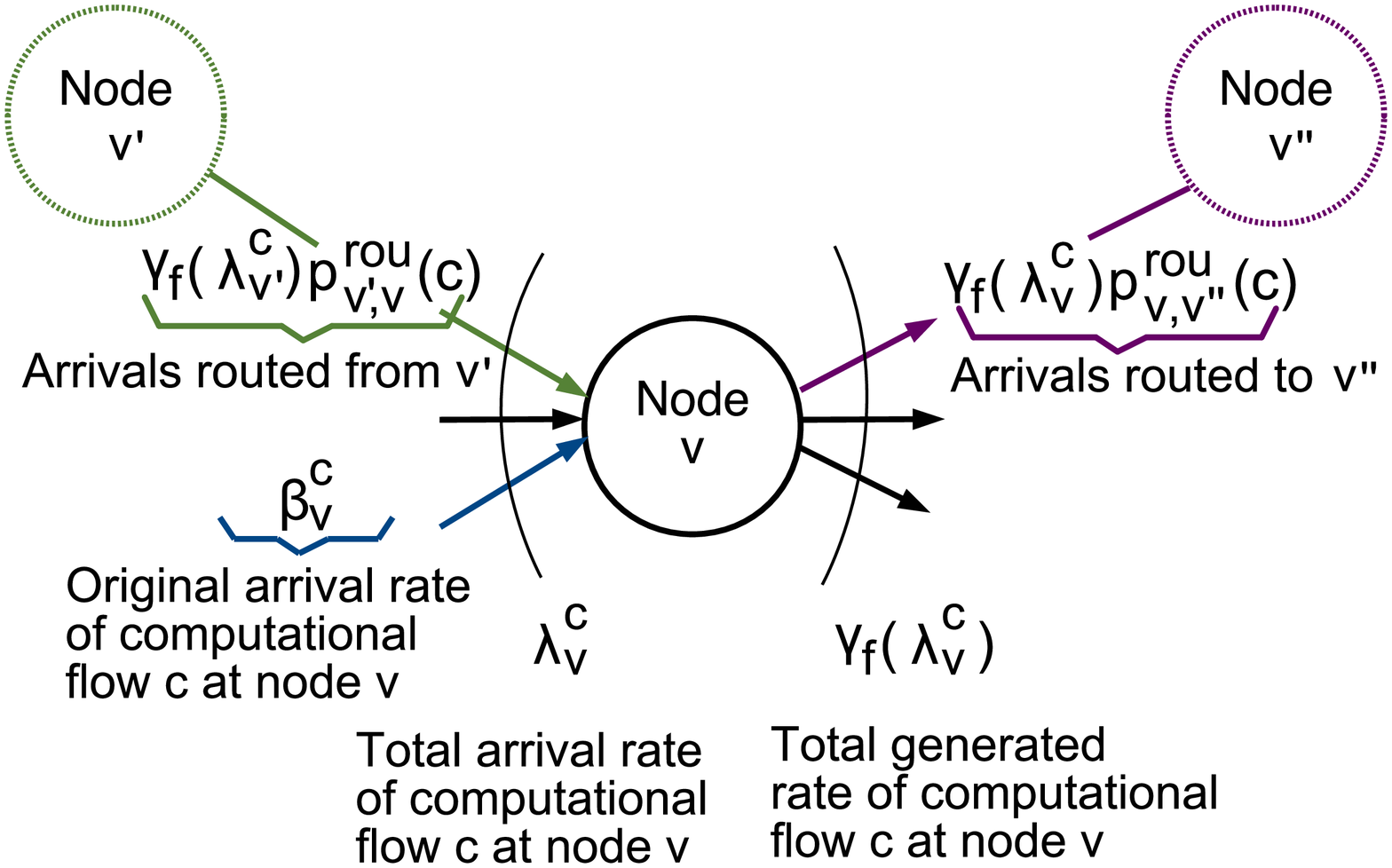}
\caption{Illustration of computational flow at node $v\in V$ where min cut $\lambda_v^c$ denotes the total arrival rate of computational flow $c$ at node $v$ that incorporates the original arrivals $\beta_v^c$ and the arrivals routed from $v'\in V$. Min cut $\gamma_f(\lambda_v^c)$ denotes the total generated rate of $c$ at $v$. Nodes $v',v''\in V$ represent the nodes where arrivals routed from/to.}
\label{flowpernode}
\end{figure} 
 	
We consider three different function categories and the time complexity of these. For search function that tries to locate an element in a sorted array, an algorithm runs in logarithmic time, which has low complexity. For MapReduce (or linear reduce) function, since the reduce functions of interest are linear, the algorithm runs in linear time, which is of medium complexity. For classification function, we can consider the set of all decision problems that have exponential runtime, which is of high complexity. 
The time complexity, i.e. the order of the count of operations, of these functions satisfies:
\begin{align}
\label{time_complexity_special_functions}
d_f(M_v^c)=\begin{cases}
O(\log(M_v^c)),\quad&\text{Search},\\
O(M_v^c),\quad&\text{MapReduce},\\
O(\exp(M_v^c)),\quad&\text{Classification}.\\
\end{cases}
\end{align}
In Sect. \ref{performance}, we evaluate and contrast $\rho_{th}$ for above functions. 

\subsection{Routing for Computing}

	We assume an open network and that the arrival rate of class $c$ packets to the system is Poisson with rate $\beta^c$. Let $p_v^{arr}(c)$ be the probability that an arriving class $c$ packet is routed to queue $v$. Assuming that all arriving packets are assigned to a queue, we have that $\sum\nolimits_{v\in V}p_v^{arr}(c)=1$.
	
	For tractability, we consider the behavior of each node in isolation. This is allowed given that the network is quasi-reversible or product form \cite{walrand1983probabilistic}. For example, a Jackson network exhibits this behavior. With this, we assume a Markov routing policy \cite[Ch. 10.6.2]{nelson2013probability} which can be described as follows. As a result of function computation, packets might have different classes, and routing probabilities depend on a packet's class. However, we assume that packets do not change their class when routed from one node to another. 
	Let $p_{v,v'}^{rou}(c)$ be the probability that a class $c$ packet that finishes service at node $v$ is routed to node $v'$. The probability that a class $c$ packet departs from the network after service completion at node $v$ is given by $p_v^{dep}(c)=1-\sum\nolimits_{v'\in V}p_{v,v'}^{rou}(c)$, where the second term on the RHS denotes the total probability that the packets stay in the network. Since it is an open network model, 
	for every class $c$ there is at least one value of $v$ so that $p_v^{dep}(c) > 0$. Thus all packets eventually leave the system. 
	
	For simplicity, assume that conversion among classes is not possible\footnote{In general packets can change their class when routed from one node to another \cite{nelson2013probability}. The study of the multi-class generalization is left as future work. 
	}. Then the total arrival rate of class $c$ packets to $v$ is 
\begin{align}
	\label{relation_lambdav_lambdavout}
	\lambda_v^c = \beta_v^c + \sum\nolimits_{v'\in V}\gamma_f(\lambda_{v'}^{c})p_{v',v}^{rou}(c),	\end{align}
	where $\beta_v^c=\beta^c p_v^{arr}(c)$, and the first term on the RHS denotes the original arrival rate of class $c$ packets that are assigned to node $v$, and the second term on the RHS denotes the arrival rate of class $c$ packets that are routed to node $v$ after finishing service at other nodes $v'\in V$. Note that the term $\gamma_f(\lambda_{v'}^{c})$ denotes the total departure rate of class $c$ packets from node $v'$ (as a result of computation). Furthermore, the total departure rate of class $c$ packets from $v$ in the forward process is given by $\lambda_v^c p_v^{dep}(c)$. Let $p_{v}^{rou}=[p_{v,v'}^{rou}(c)]_{v'\in V,\,c\in\mathcal{C}}\in\mathbb{R}^{V\times C}$. 
	
	Advantages of having a Jackson type network as in (\ref{relation_lambdav_lambdavout}) are such that nodes can be considered in isolation. Each node needs to know how much it needs to manage, which is less complicated than when nodes need the topological information to determine how to manage individual computational flows.
	
\begin{figure}[t!]
\centering
\includegraphics[width=0.32\textwidth]{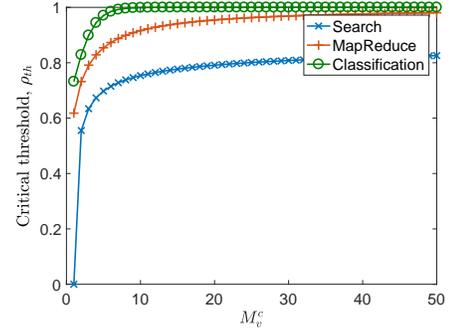}
\caption{The critical threshold $\rho_{th}$ for computation versus $M_v^c$ .}  
\label{thresholdisolated}
\end{figure}

\subsection{Solution to MinCost Problem}	
Using (\ref{relation_lambdav_lambdavout}), we rewrite the MinCost formulation in (\ref{costoptimization}) as
	\begin{equation}
	\label{costoptimization_expanded}
	\begin{aligned}
	\hspace{-0.6cm}{\rm MinCost}:\!\! & \underset{\gamma_f<\lambda<\mu}{\min}\!\!
	& & \!\!\!\!C \!=\! \sum\limits_{v\in V} \sum\limits_{c\in\mathcal{C}} \frac{d_f(M_v^c)}{\lambda_v^c}\!+\!\frac{1}{\mu_v^c-\gamma_f(\lambda_v^c)}\\	
	& \hspace{0.3cm}\text{s.t.}
	& & \lambda_v^c< \mu_v^c, \quad \forall c\in \mathcal{C},\,\, v\in V,
	\end{aligned}
	\end{equation}	
	where we assume that $p_{v}^{rou}(c)$, $\beta_v^c=\beta^c p_v^{arr}(c)$ are known apriori, and from (\ref{flow_entropy_relation}), $\gamma_f(\lambda_v^c)$ satisfies $\forall c\in \mathcal{C},\,\, v\in V$ that
	\begin{align}
	\label{flow_entropy_relation_expanded}
	\gamma_f(\lambda_v^c)
	\geq \Big[\beta_v^c + \sum\limits_{v'\in V}\gamma_f(\lambda_{v'}^{c})p_{v',v}^{rou}(c)\Big]\Gamma_c,
	\end{align} 
	which follows from (\ref{relation_lambdav_lambdavout}). We rewrite (\ref{flow_entropy_relation_expanded}) in vector form:
	\begin{align}
	\label{flow_entropy_relation_vector}
	{\bm \gamma}_f({\bm \lambda}^c) \geq {\bm \gamma}_{f,\, LB}({\bm \lambda}^c)=\left[{\bm \beta}^c+P^{rou}(c)  {\bm \gamma}_f({\bm \lambda}^c)\right]\Gamma_c, 
	\end{align}
	where $P^{rou}(c)=[p_{v',v}^{rou}(c)]_{v'\in V,v\in V}$, ${\bm \beta}^c=[\beta_v^c]_{v\in V}\in\mathbb{R}^{V\times 1}$, and ${\bm \lambda}^c=[\lambda_v^c]_{v\in V}\in\mathbb{R}^{V\times 1}$. Using the vector notation ${\bm \gamma}_{f}({\bm \lambda}^c)=[\gamma_f(\lambda_v^c)]_{v\in V}\in\mathbb{R}^{V\times 1}$, and from  (\ref{flow_entropy_relation_vector}) we obtain	
	\begin{align}
	\label{gamma_lower_bound}
		{\bm \gamma}_{f,\, LB}(\lambda^c)=\left(I-P^{rou}(c)\Gamma_c\right)^{-1}{\bm \beta}^c \Gamma_c, \quad \forall c\in \mathcal{C},
		\end{align}
		where $I$ is an $V\times V$ identity matrix.
	To guarantee that ${\bm \gamma}_{f}({\bm \lambda}^c)<{\bm \lambda}^c$, we can use the above condition in (\ref{gamma_lower_bound}). Hence, a necessary condition for the external arrival rate ${\bm \beta}^c$ for the computation operation to be effective is given by
	\begin{align}
	{\bm \beta}^c &<\left(I-P^{rou}(c)\Gamma_c\right) {\bm \lambda}^c /\Gamma_c.\nonumber
	\end{align}
	
\begin{figure*}[t!]	
	\centering
	\includegraphics[width=0.245\textwidth]{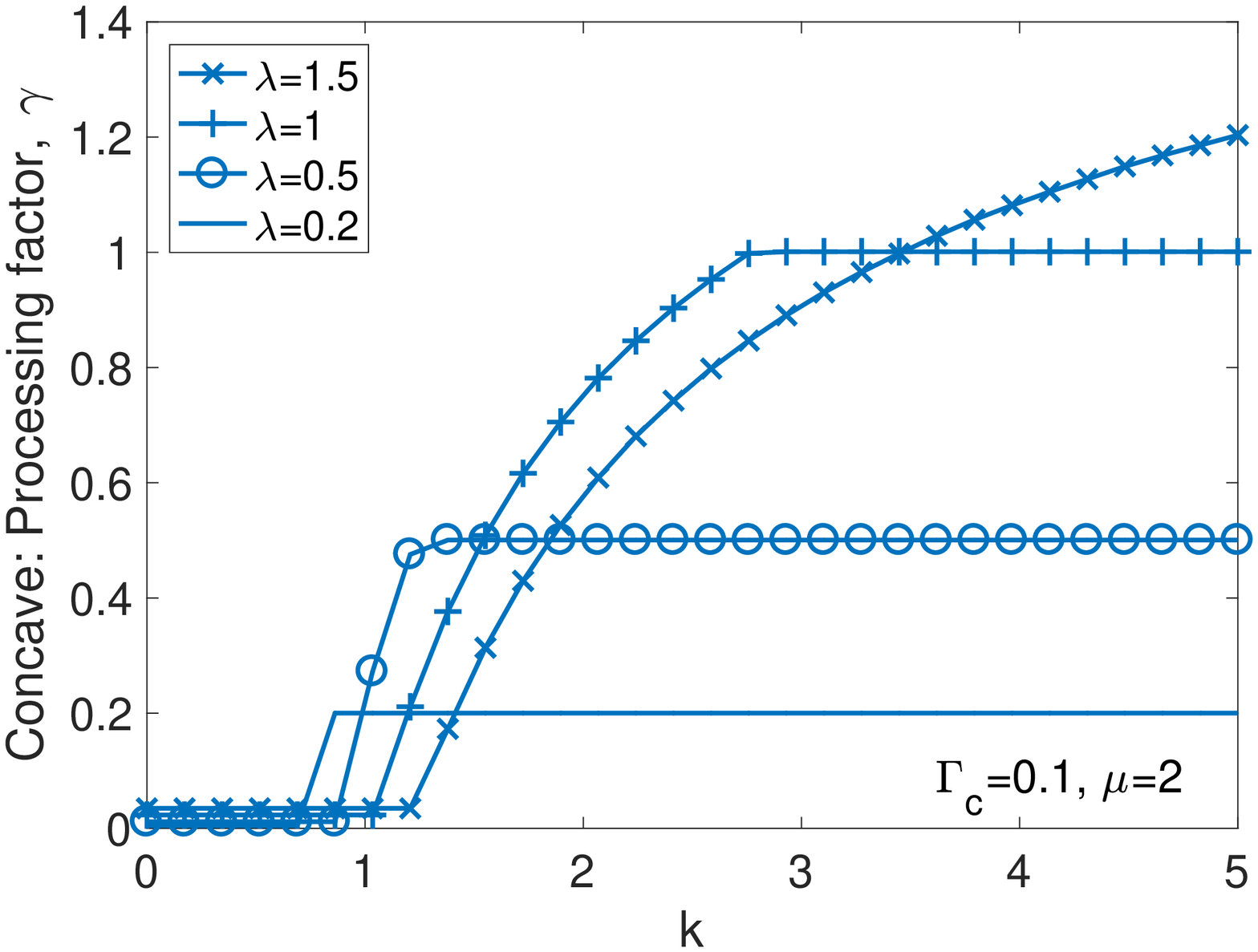}
	\includegraphics[width=0.245\textwidth]{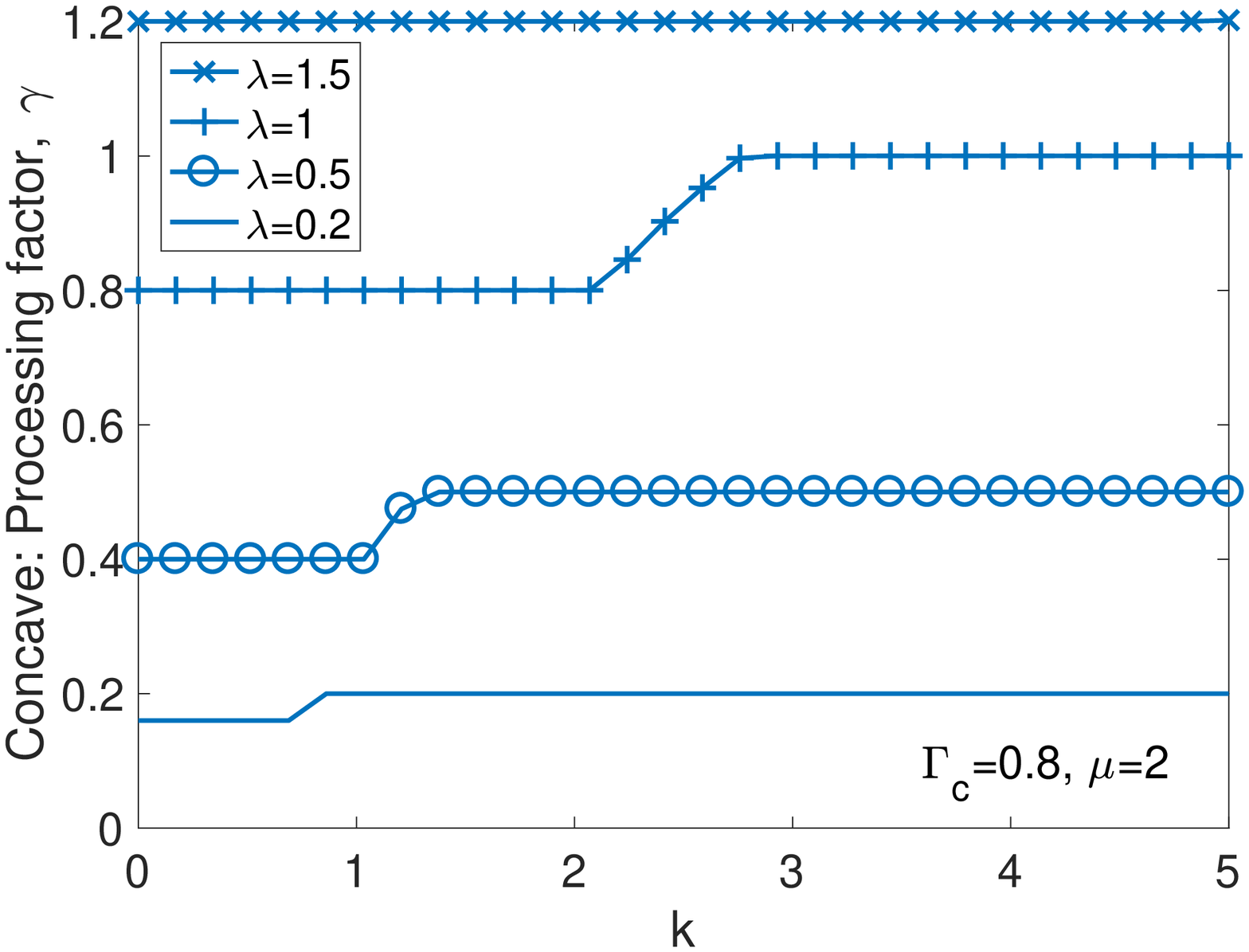}
	\includegraphics[width=0.245\textwidth]{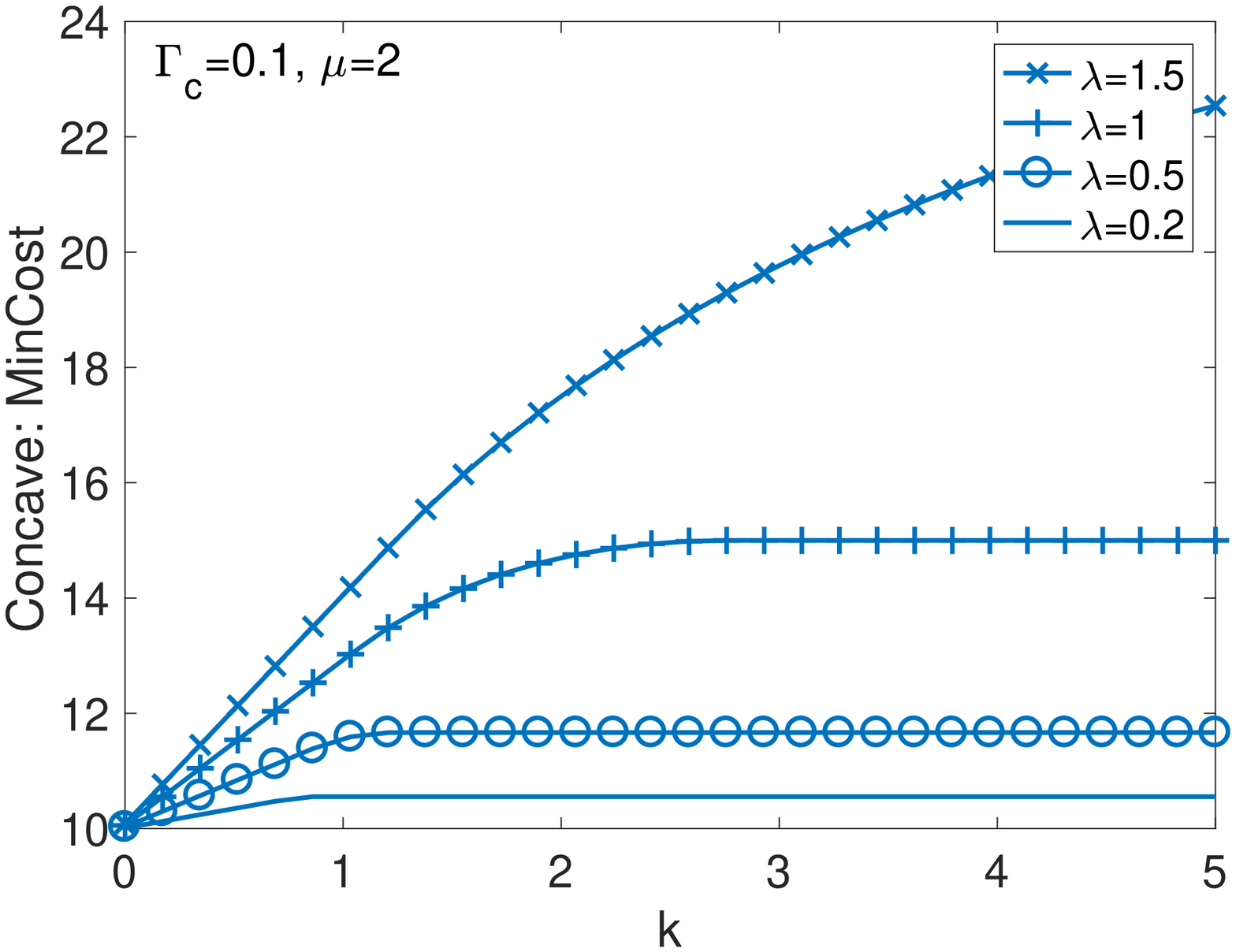}
	\includegraphics[width=0.245\textwidth]{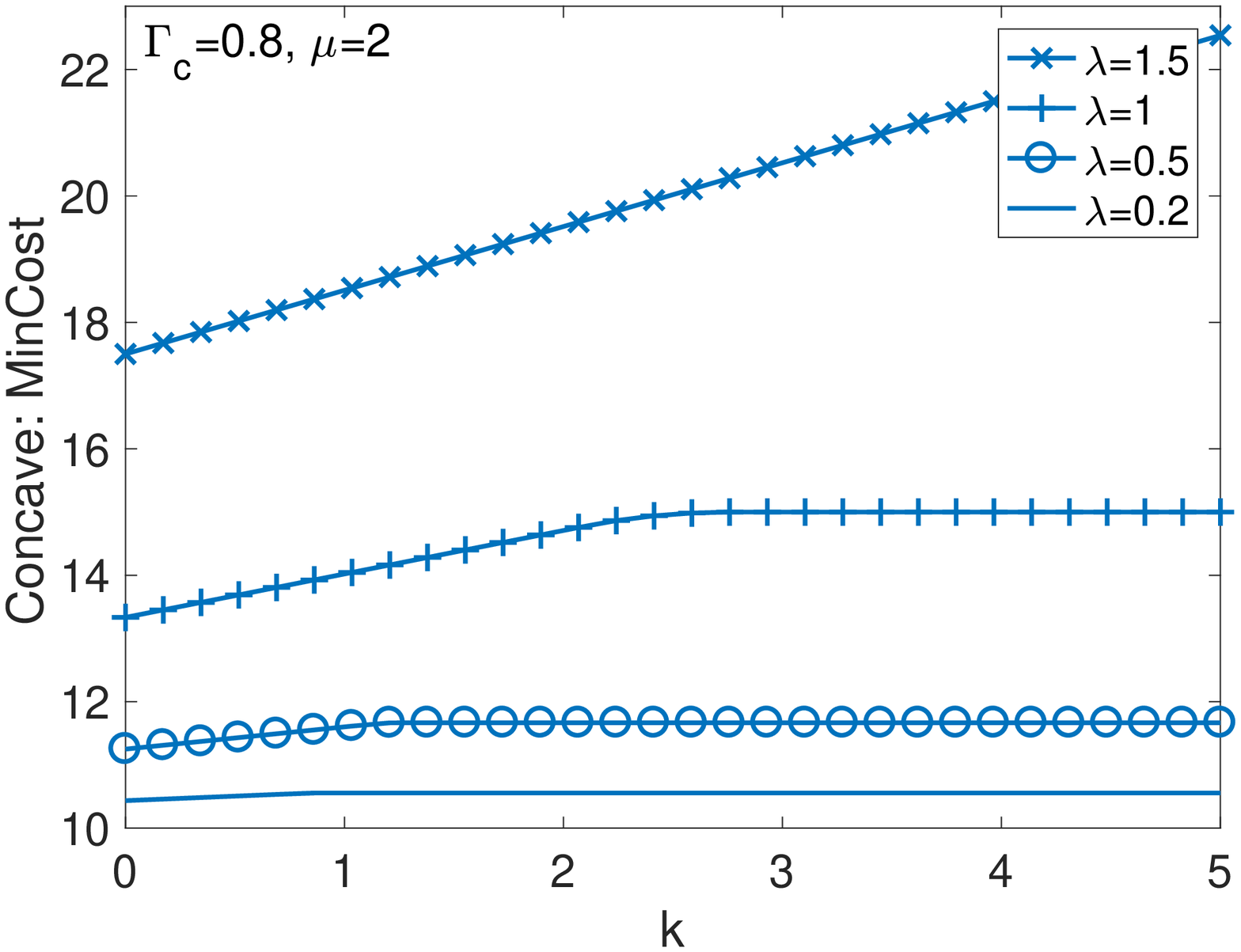}
	\caption{Concave computation cost. (L) Processing factor $\gamma_f(\lambda_v^c)$ versus computation cost scaling factor $k$. (R) MinCost versus $k$.}
	\label{Concave_MinCost_k}
	\end{figure*}
	
	\begin{figure*}[t!]	
	\centering
	\includegraphics[width=0.245\textwidth]{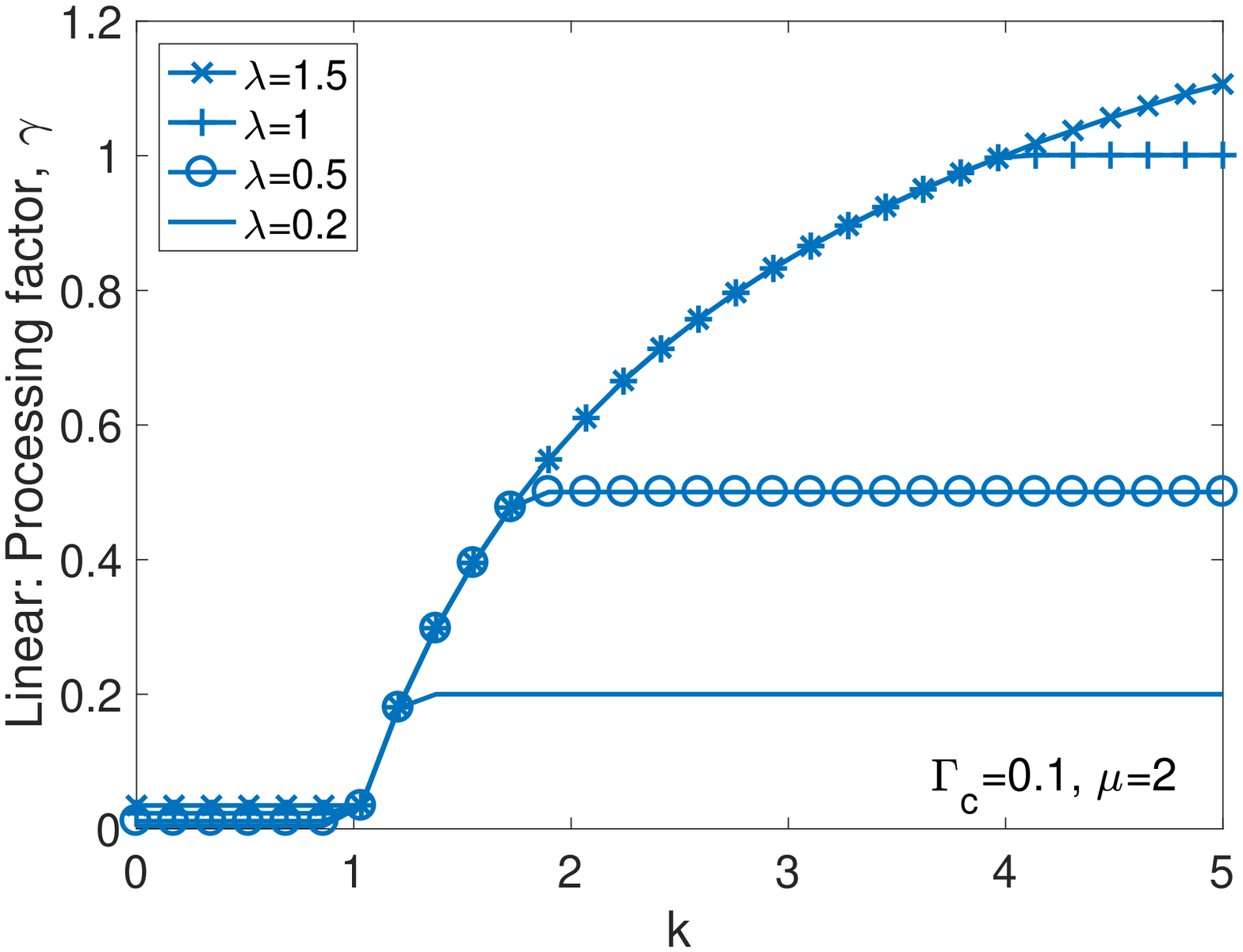}
	\includegraphics[width=0.245\textwidth]{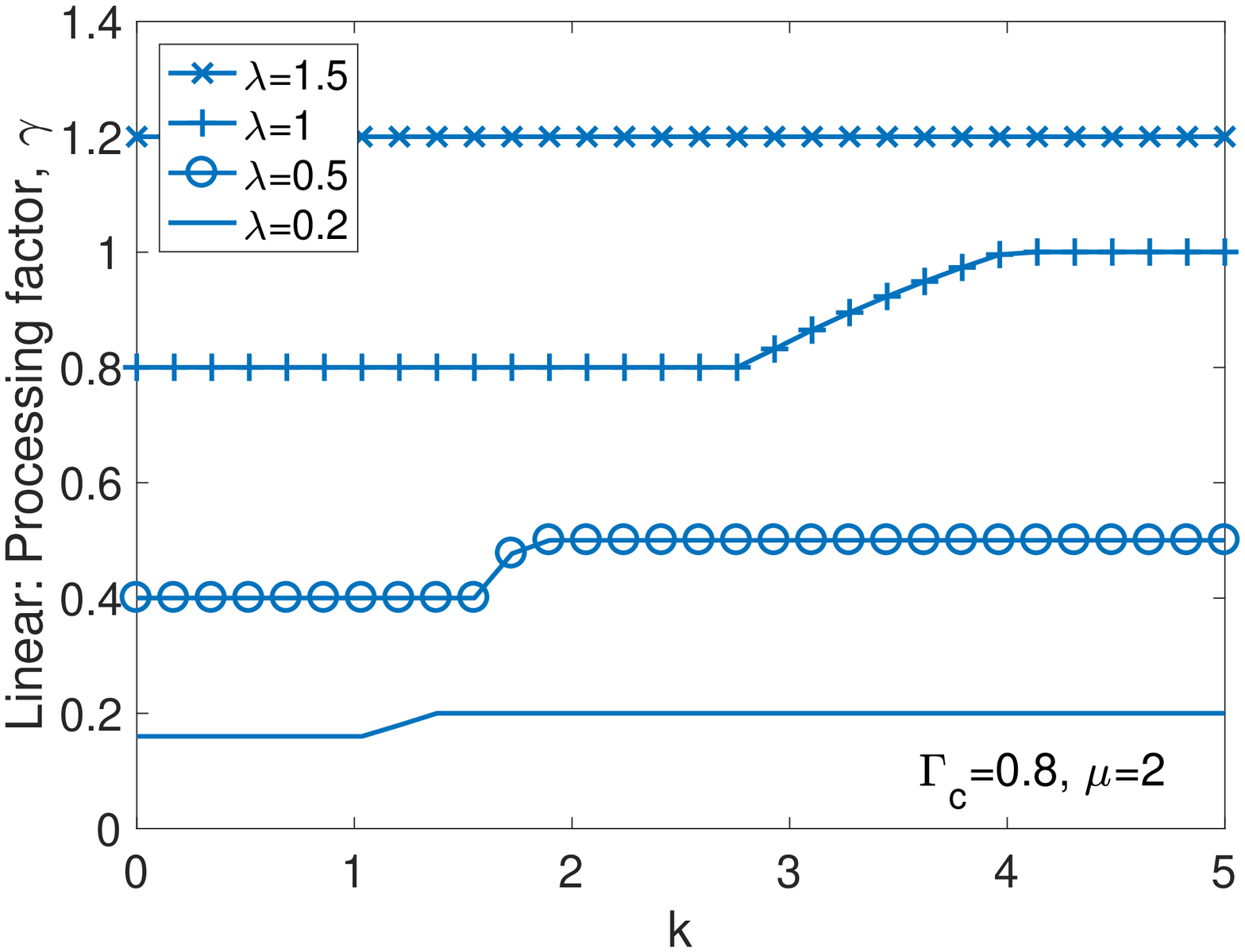}
	\includegraphics[width=0.245\textwidth]{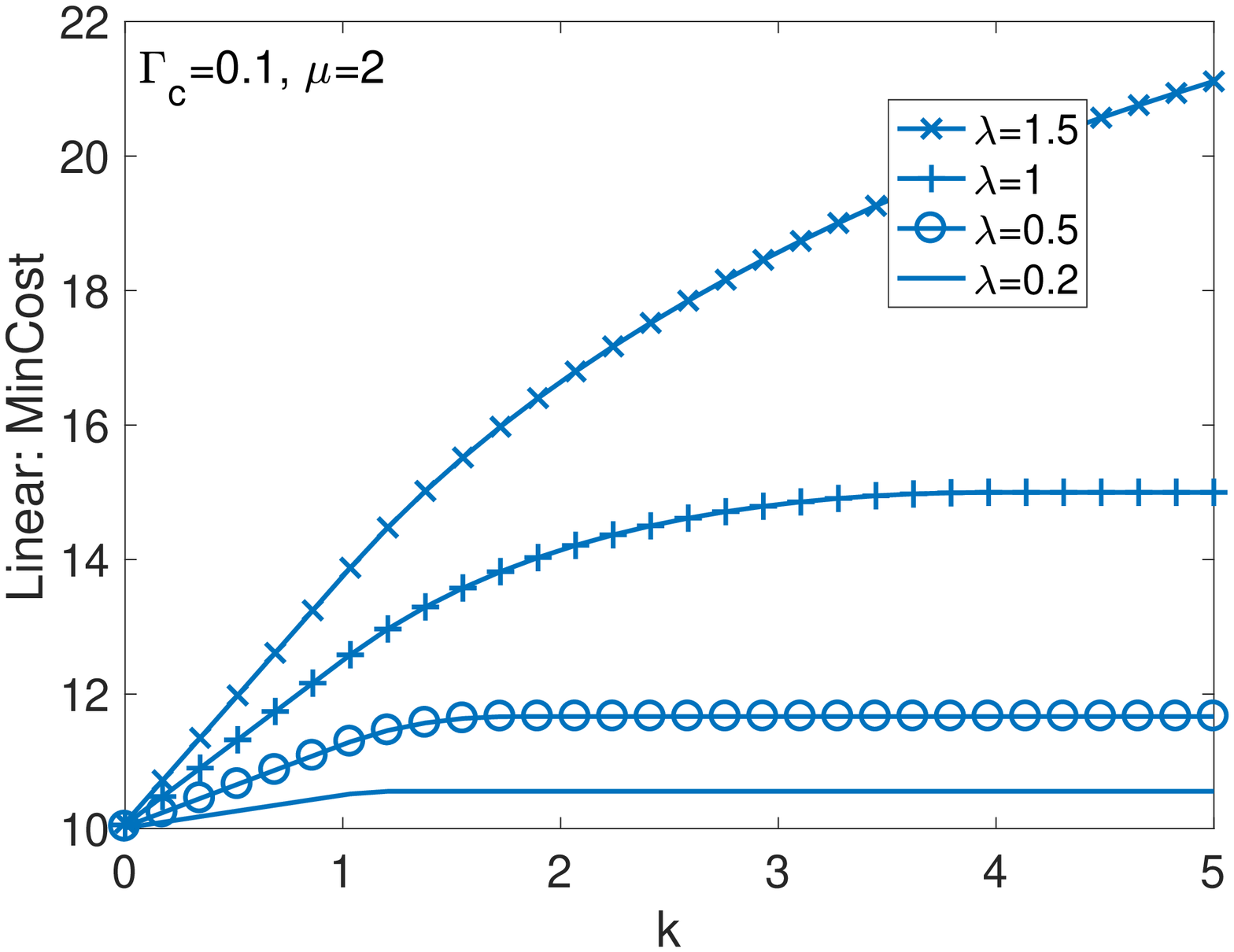}
	\includegraphics[width=0.245\textwidth]{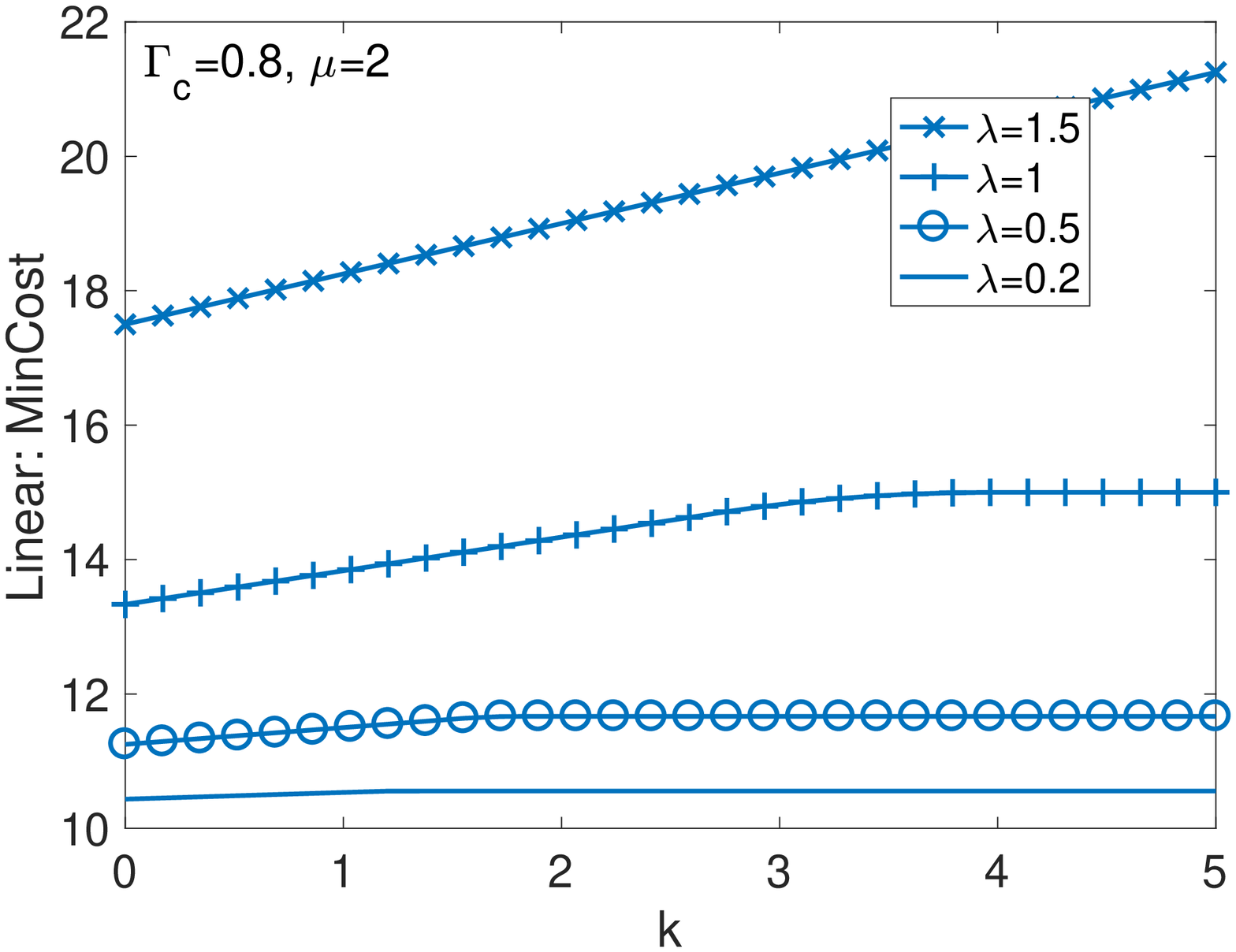}
	\caption{Linear computation cost. (L) Processing factor $\gamma_f(\lambda_v^c)$ versus computation cost scaling factor $k$. (R) MinCost versus $k$.}
	\label{Linear_MinCost_k}
	\end{figure*}	

The range of ${\bm \lambda}^c$ can be computed as function of $\Gamma_c$ as
	\begin{align}
	\label{lambda_interval}
	(I-P^{rou}(c)\Gamma_c)^{-1} {\bm \beta}^c \leq {\bm \lambda}^c<(I-P^{rou}(c))^{-1}{\bm \beta}^c.
	\end{align}	
	Note that $\Gamma_c$ is determined by the function's surjectivity hence is is affected by the computation cost of the function. Given $\Gamma_c$, (\ref{gamma_lower_bound}) gives a lower bound on ${\bm \gamma}_{f}({\bm \lambda}^c)$. Furthermore, using the relation ${\bm \lambda}^c={\bm \beta}^c+P^{rou}(c)  {\bm \gamma}_f({\bm \lambda}^c)<{\bm \beta}^c+P^{rou}(c)  {\bm \lambda}^c$ in (\ref{relation_lambdav_lambdavout}), we  equivalently require that ${\bm \lambda}^c<(I-P^{rou}(c))^{-1}{\bm \beta}^c$, yielding the upper bound in (\ref{lambda_interval}). Furthermore, we have ${\bm \lambda}^c={\bm \beta}^c+P^{rou}(c)  {\bm \gamma}_f({\bm \lambda}^c)\geq {\bm \beta}^c+P^{rou}(c)  {\bm \lambda}^c\Gamma_c$, and hence, ${\bm \lambda}^c\geq (I-P^{rou}(c)\Gamma_c)^{-1} {\bm \beta}^c$, yielding the lower bound in (\ref{lambda_interval}). 
	
	In (\ref{costoptimization_expanded}), the time complexity of classes, i.e. $d_f(M_v^c)$, is known. Observe that $C_{v,comp}^c$ decreases in $\gamma_f(\lambda_{v}^{c})$, and $C_{v,comm}^c$ increases in $\gamma_f(\lambda_{v}^{c})$. Note also that due to (\ref{flow_entropy_relation_expanded}) the values of $\lambda_{v}^{c}$ should be jointly optimized to minimize $C$.

In (\ref{costoptimization}), given the cost functions $C_{v,comm}^c$ and $C_{v,comp}^c$, we can solve for the optimal values of $\gamma_f(\lambda_v^c)$, $v\in V$, $c\in\mathcal{C}$ that minimize the MinCost problem. Then, using the entropic surjectivity relation in (\ref{flow_entropy_relation}), and by mapping the surjectivity to the class of functions, we can infer the type of flows (i.e. functions) that we can compute effectively. Using the order of the count of operations given in (\ref{time_complexity_special_functions}), we model the delay cost functions for computations of different classes of functions:  
	\begin{enumerate}
		\item Search (concave): \!$C_{v,comp}^c\!\!=\!\! \frac{1}{\mu_v^c}\Big(1\!+k\log\big(1\!+\!\frac{\lambda_v^c-\gamma_f(\lambda_v^c)}{\mu_v^c}\big)\Big)$, 
		\item	MapReduce (linear): $C_{v,comp}^c= \frac{1}{\mu_v^c}\Big(1+k\frac{\lambda_v^c-\gamma_f(\lambda_v^c)}{\mu_v^c}\Big)$, 
		\item Classification (convex): $C_{v,comp}^c= \frac{1}{\mu_v^c-k(\lambda_v^c-\gamma_f(\lambda_v^c))} $, 
		\end{enumerate}
		where $k$ is some constant as a proxy for the cost. In Sect. \ref{performance}, we numerically investigate the behavior of MinCost with respect to $k$. 	Note that above models satisfy $C_{v,comp}^c=\frac{1}{\mu_v^c}$ when $\gamma_f(\lambda_v^c)=\lambda_v^c$. Furthermore, if $k\to 0$, the computation cost is $C_{v,comp}^c=\frac{1}{\mu_v^c}$, which is not affected by $\gamma_f(\lambda_v^c)$. Note also that $C_{v,comp}^c$ decreases in $\gamma_f(\lambda_v^c)$, then due to (\ref{costoptimization_expanded}), there is a value of $\gamma_f(\lambda_v^c)$ that optimizes MinCost $\forall c\in\mathcal{C}$, $v\in V$.

\begin{ex}\label{ex_convex}
{\bf Modeling classification via convex flow.} We plug the $C_{v,comp}^c$ expression into the MinCost formulation in (\ref{costoptimization_expanded}), and then use (\ref{flow_entropy_relation_vector}) to compute ${\bm \gamma}_f(\lambda_v^c)$ for given set of $\lambda_v^c$'s.	
	Since the objective function is convex, i.e. $\frac{\partial^2 C}{\partial (\gamma_f(\lambda_v^c)) ^2}>0$, the optimal solution can be found by solving $\frac{\partial C}{\partial \gamma_f(\lambda_v^c)}=0$ as $\gamma_f(\lambda_v^c)=\lambda_v^c/2$, $\forall v$, $\forall c$. 
Hence, we can decide $\gamma_f(\lambda_v^c)$ values using the solution of (\ref{costoptimization_expanded}) and from the set of $V\times C$ equalities with $V\times C$ unknown $\lambda_{v}^{c}$ values due to (\ref{flow_entropy_relation_vector}) as given below: 
	\begin{align}
	\frac{\lambda_{v}^{c}}{2}=\Big[\beta_v^c + \sum\limits_{v'\in V}\frac{\lambda_{v'}^{c}}{2}p_{v',v}^{rou}(c)\Big]\Gamma_c,\quad \forall c\in \mathcal{C},\,\, v\in V,
	\end{align}
	using which we get ${\bm \lambda}^c= 2(I-P^{rou}(c)\Gamma_c)^{-1}{\bm \beta}^c \Gamma_c$, $\forall c\in \mathcal{C}$.
\end{ex}	

Since linear flow is a special case of convex flow, the optimal solution is found by solving $\frac{\partial C}{\partial \gamma_f(\lambda_v^c)}=0$ as $\gamma_f(\lambda_v^c)=\mu_v^c(1-1/\sqrt{k})$. More accurately $\frac{\partial \lambda_v^c}{\partial \gamma_f(\lambda_v^c)}=-\frac{(\mu_v^c)^2}{k(\mu_v^c- \gamma_f(\lambda_v^c))^2}+1>0$. This means that $\mu_v^c(1-1/\sqrt{k})> \gamma_f(\lambda_v^c)$. Due to space limitations we skip the discussion of the concave flow.

To find the local minima of MinCost for general computation cost functions, we use the Karush-Kuhn-Tucker (KKT) approach in nonlinear programming \cite{Boyd2009}. Allowing inequality constraints, KKT conditions determine the optimal solution:   	
	\begin{align}
	L
	= \sum\limits_{v\in V}\sum\limits_{c\in\mathcal{C}} W_v^c  + \xi_v^c (\gamma_f-\lambda_v^c)  
	+ \zeta_v^c  (\gamma_{f,\, LB}-\gamma_f),\nonumber
	\end{align}
	where $\xi=[\xi_v^c]_{v,\in V, \, c\in\mathcal{C}},\,\zeta=[\zeta_v^c]_{v,\in V, \, c\in\mathcal{C}}\geq 0$ are the dual variables, and for optimality (i.e. the solution at $\lambda=[\lambda_v^c]_{v\in V,\, c\in\mathcal{C}}=\lambda_v^{c\,*}$) the partial derivatives of $L$ satisfy
	\begin{align}
	\frac{\partial L }{\partial \gamma_f}&
	=\frac{1}{\lambda_v^{c\,*}}\frac{\partial d_f(M_v^c)}{\partial \gamma_f(\lambda_v^{c\,*})}+\frac{1}{(\mu_v^c-\gamma_f(\lambda_v^{c\,*}))^2}+\xi_v^c-\zeta_v^c=0,\nonumber\\
	\frac{\partial L }{\partial \xi_v^c}&=\gamma_f^*-\lambda_v^{c\,*}\leq 0,\quad
	\frac{\partial L }{\partial \zeta_v^c}=\gamma_{f,\, LB}^*-\gamma_f\leq 0.\nonumber
	\end{align}
	From complementary slackness, we require that 
	\begin{align}
	\xi_v^c (\gamma_f(\lambda_v^{c\,*})-\lambda_v^{c\,*}) = 0,\quad
	\zeta_v^c  (\gamma_{f,\, LB}(\lambda_v^{c\,*})-\gamma_f(\lambda_v^{c\,*})) = 0.\nonumber
	\end{align}
	The local solution of the MinCost problem is numerically derived in Sect. \ref{performance} by evaluating the above partial derivatives.

	\begin{figure*}[t!]	
	\centering
	\includegraphics[width=0.245\textwidth]{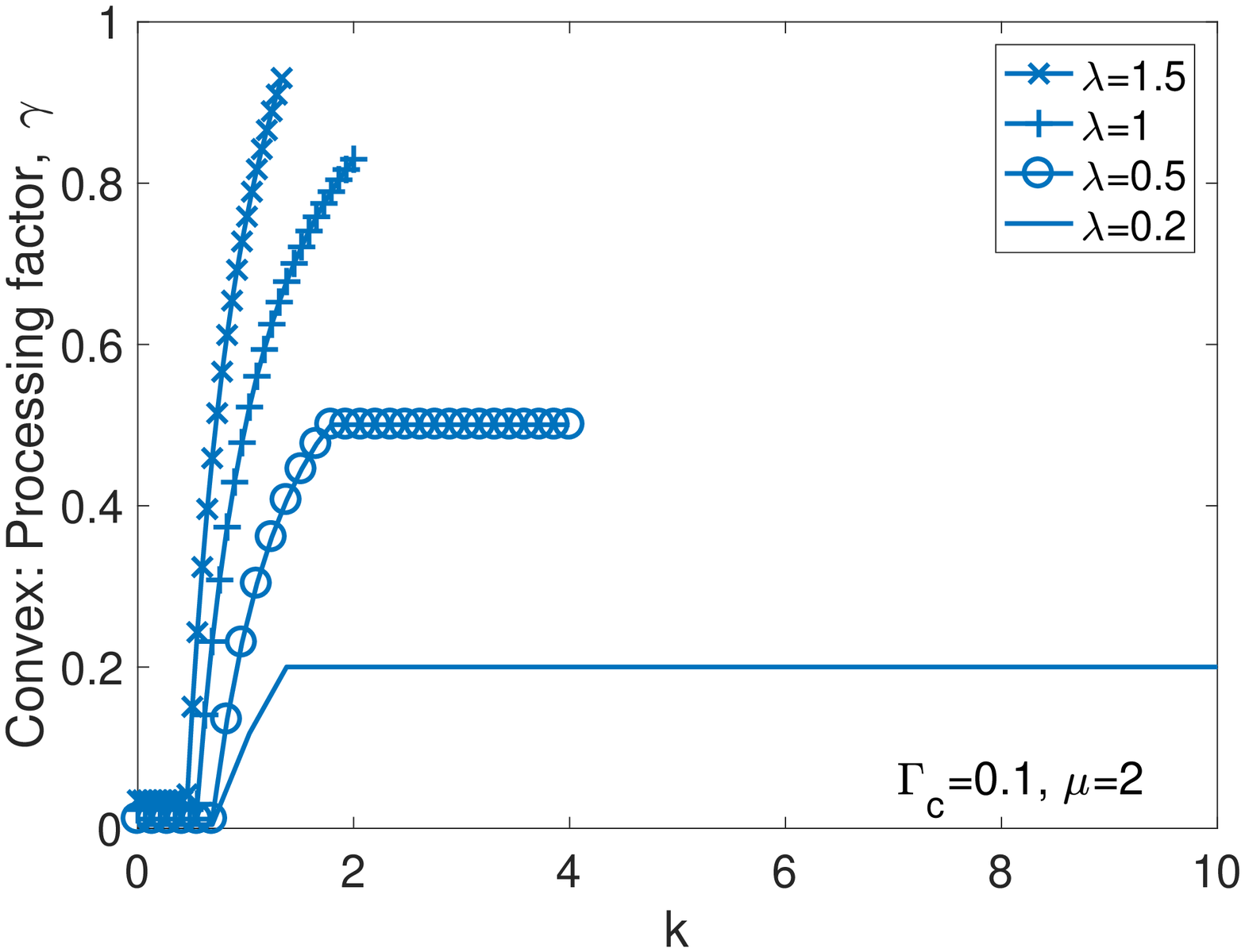}
	\includegraphics[width=0.245\textwidth]{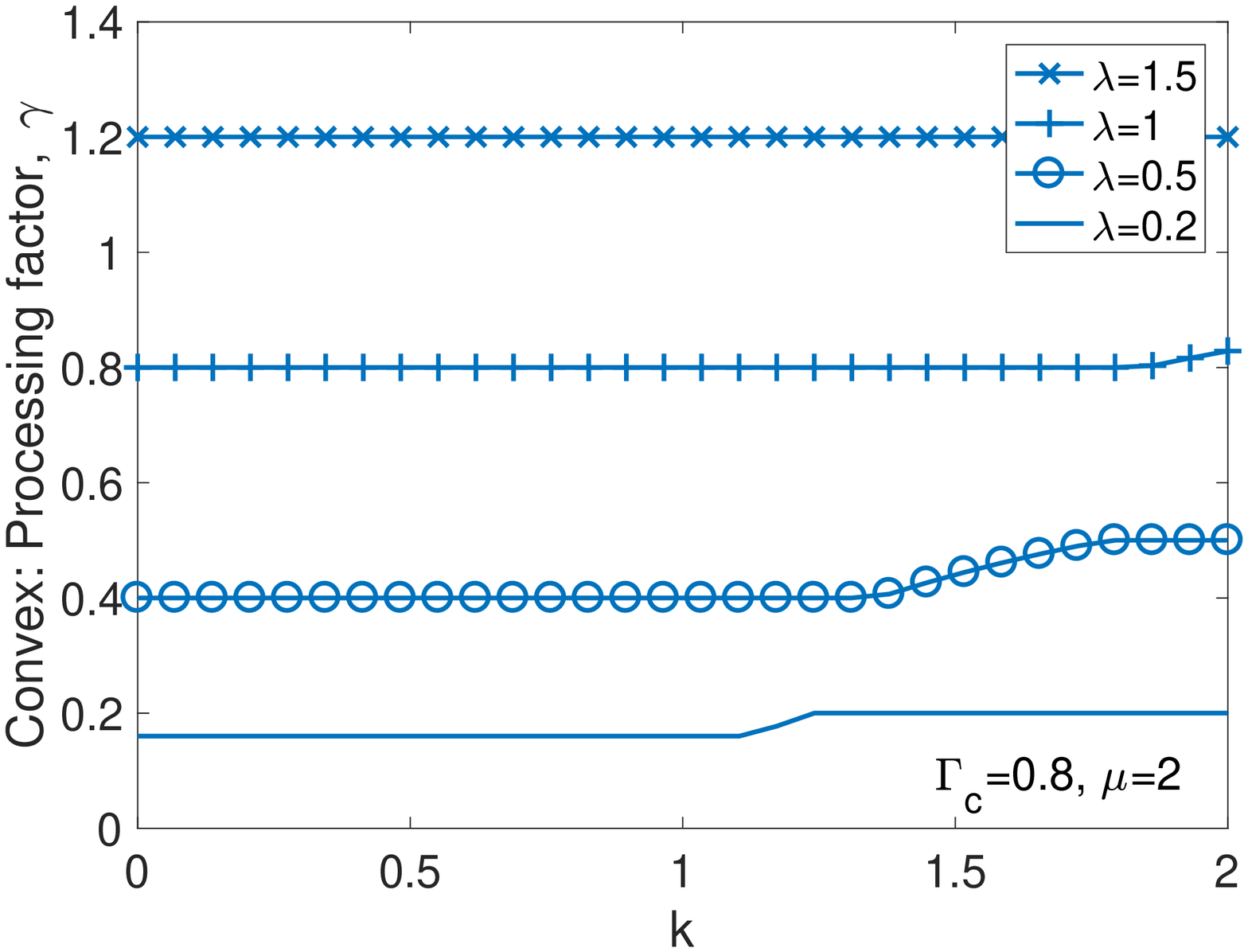}
	\includegraphics[width=0.245\textwidth]{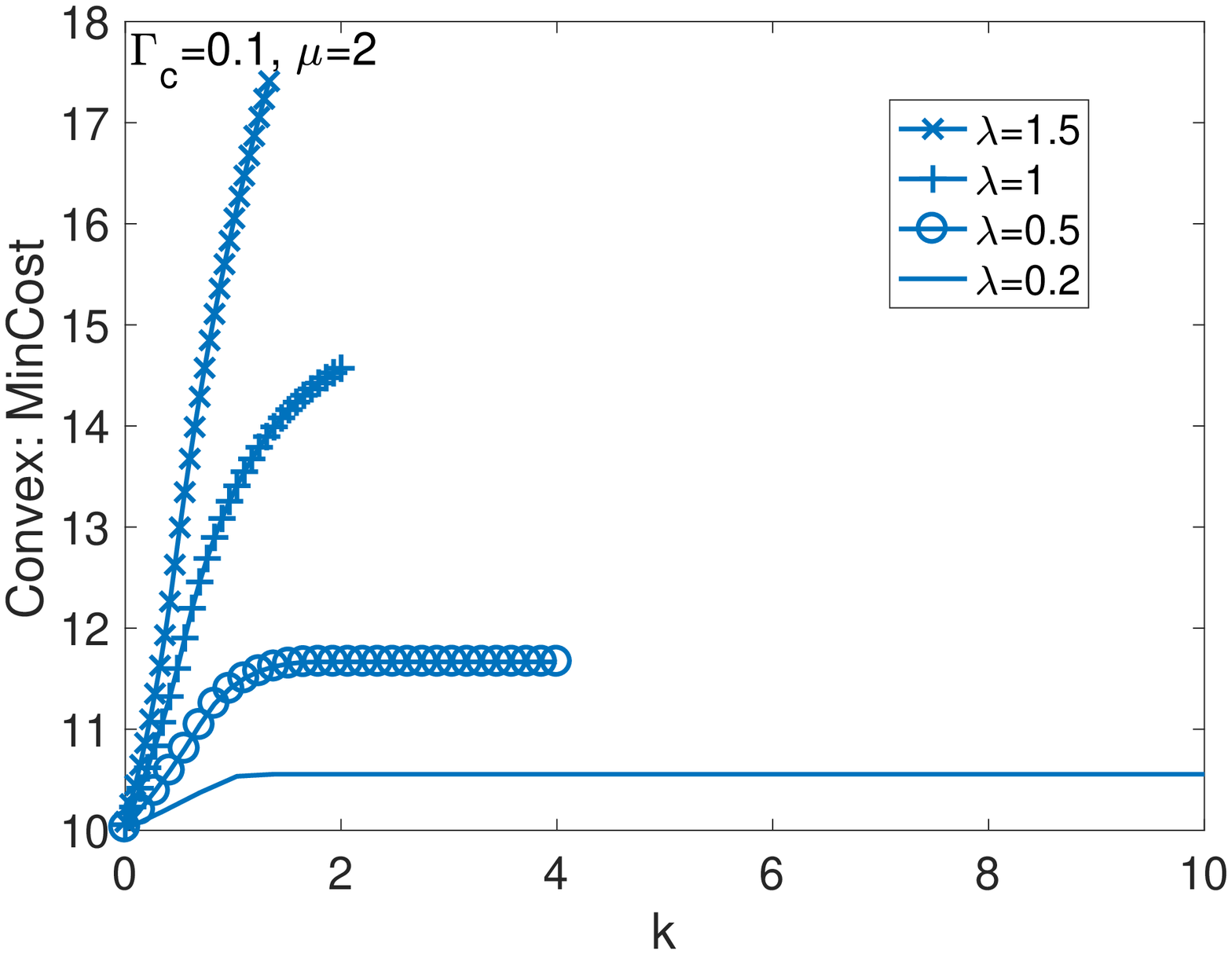}
	\includegraphics[width=0.245\textwidth]{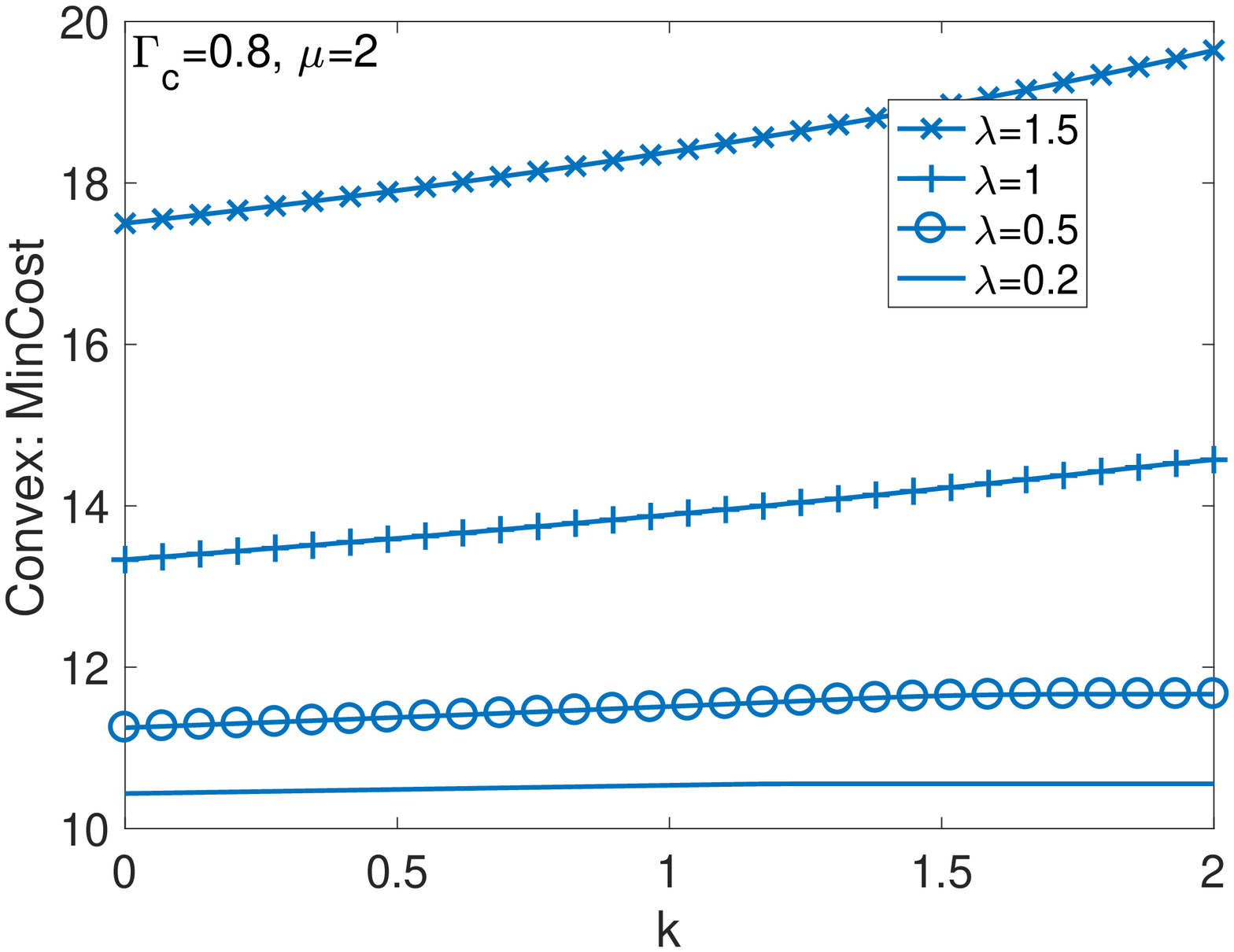}
	\caption{Convex computation cost. (L) Processing factor $\gamma_f(\lambda_v^c)$ versus computation cost scaling factor $k$. (R) MinCost versus $k$.}
	\label{Convex_MinCost_k}
	\end{figure*}	
	
		\begin{figure*}[t!]	
		\centering
		\includegraphics[width=0.245\textwidth]{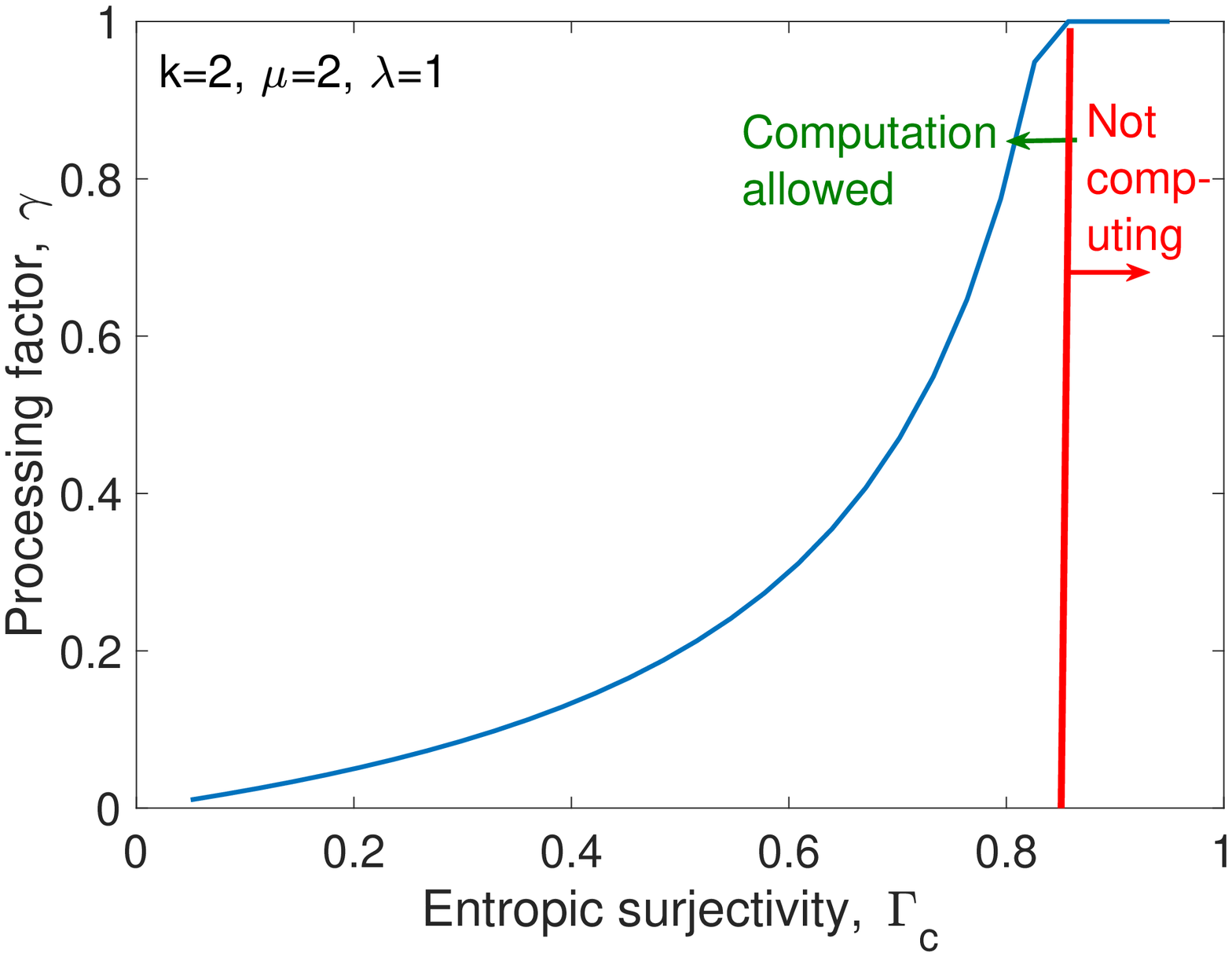}
		\includegraphics[width=0.245\textwidth]{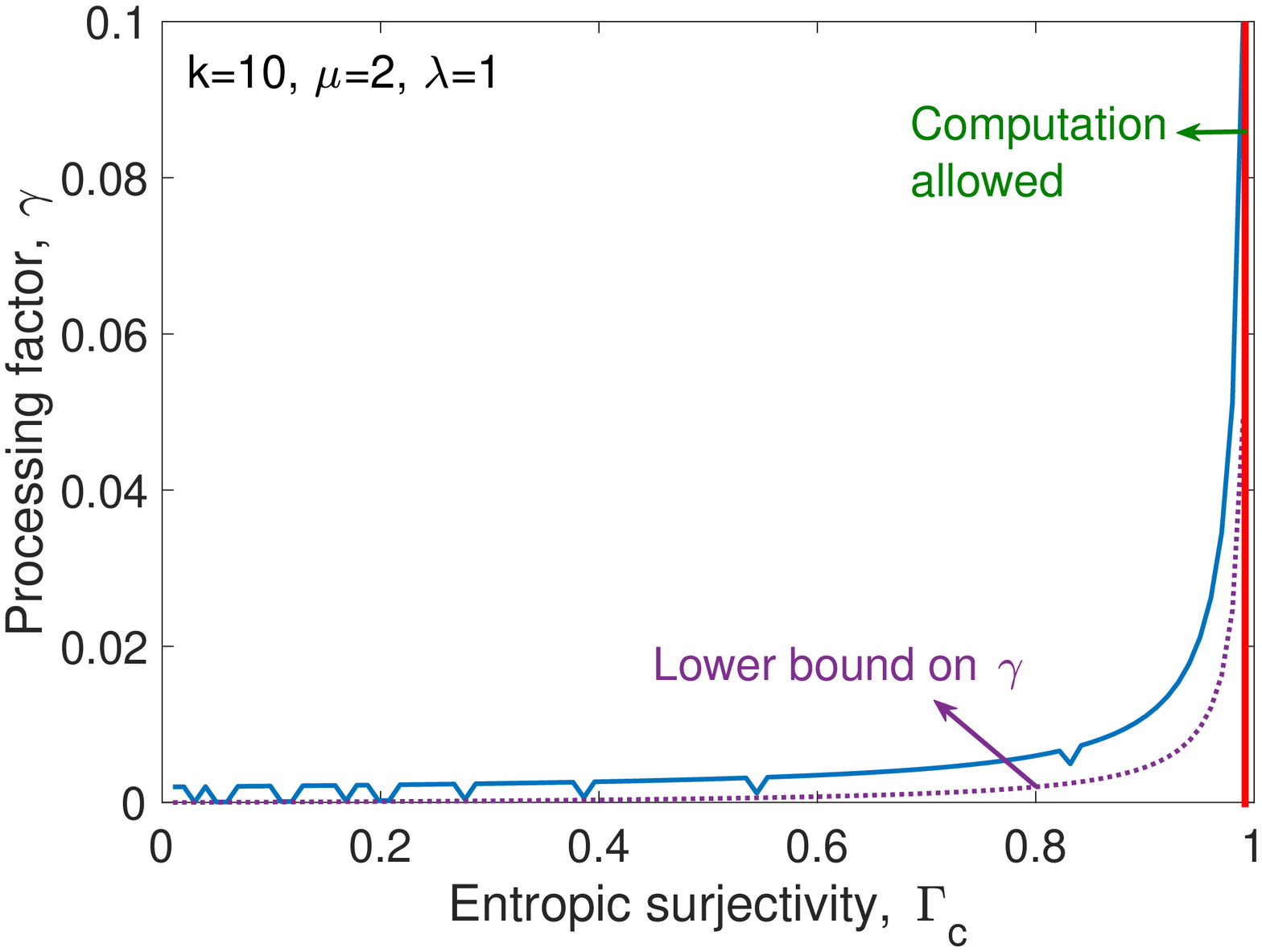}
		\includegraphics[width=0.245\textwidth]{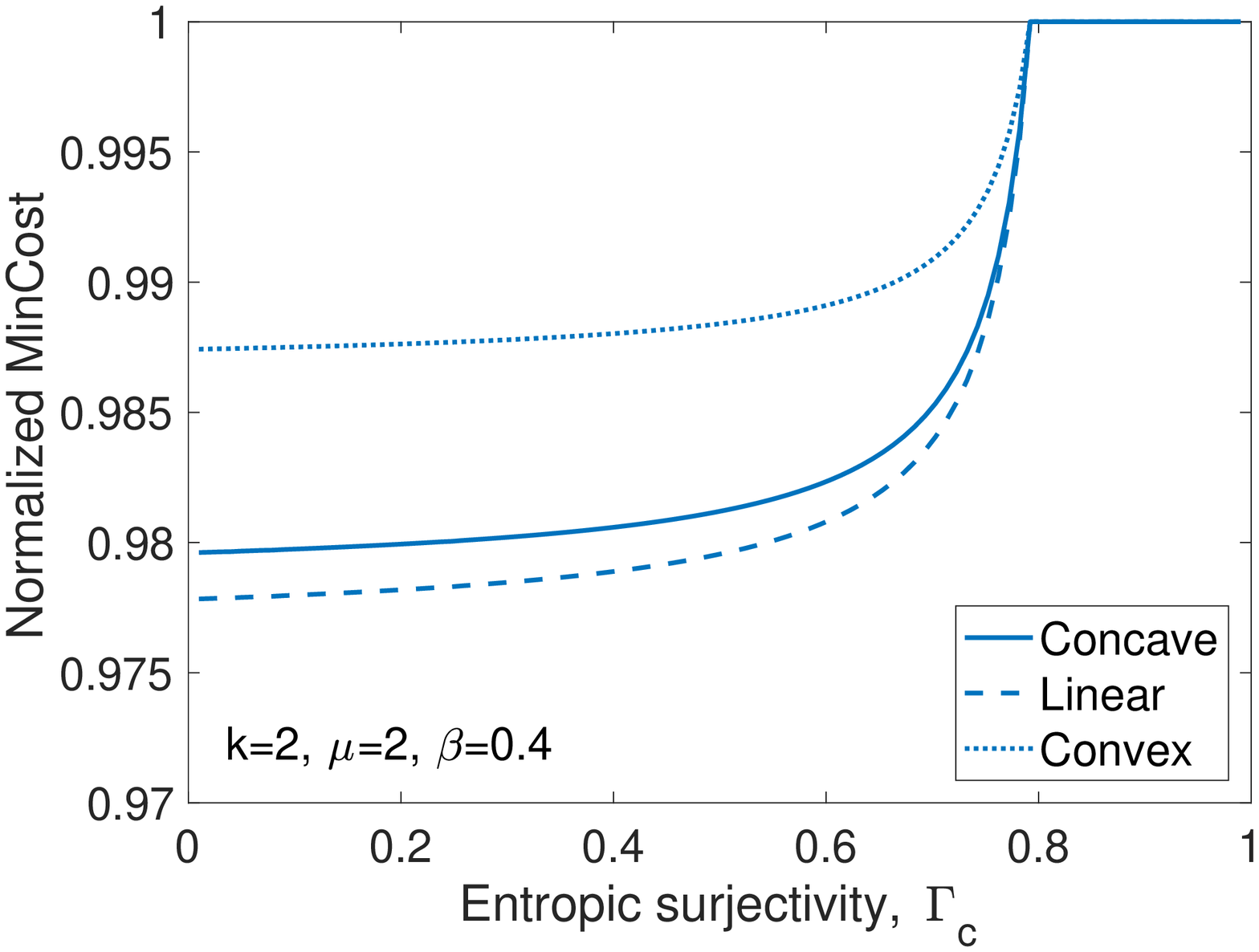}
		\includegraphics[width=0.245\textwidth]{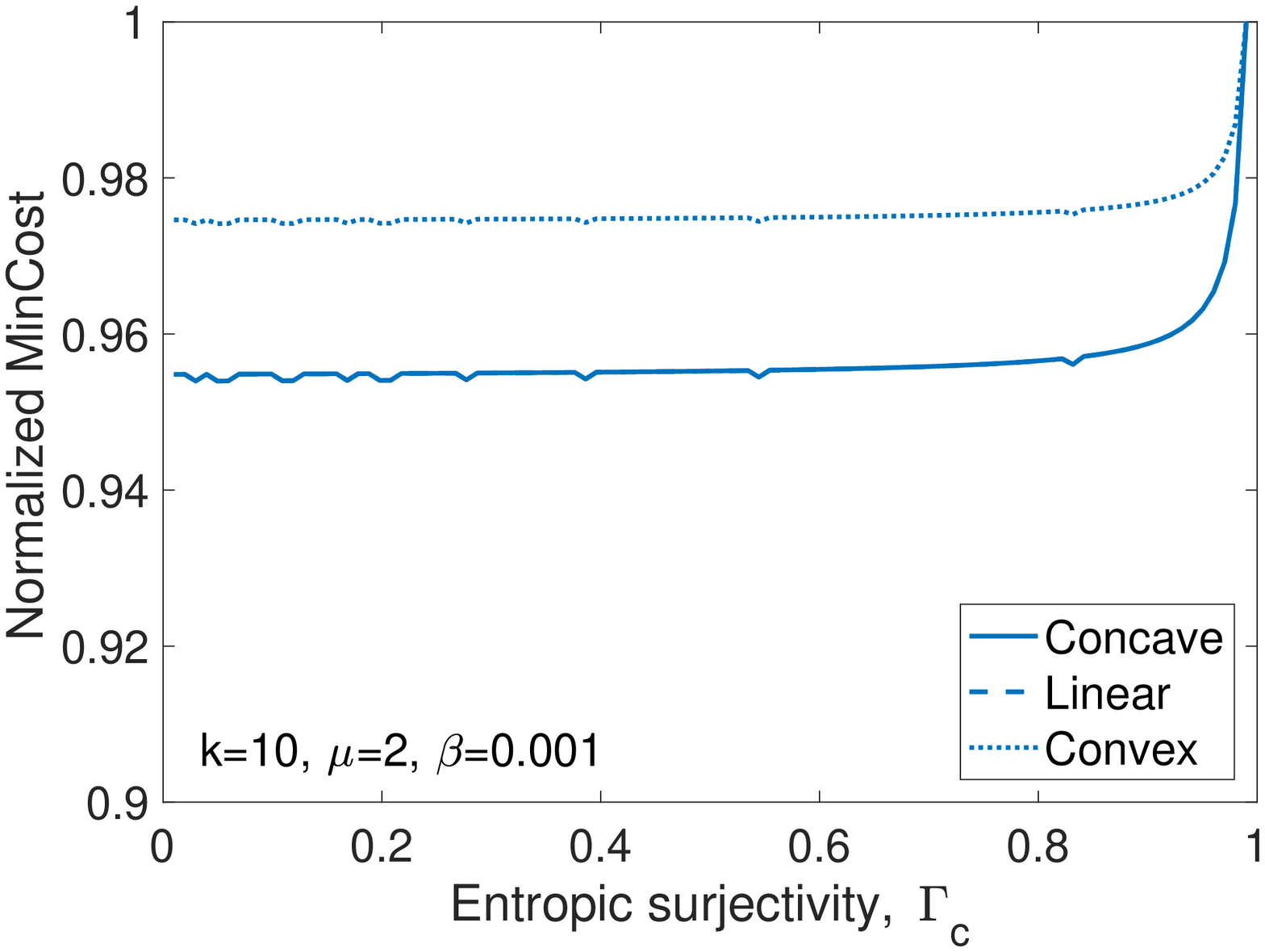}
		\caption{(L) Processing factor $\gamma_f(\lambda_v^c)$ versus entropic surjectivity $\Gamma_c$. (R) Normalized MinCost versus $\Gamma_c$ for different computation costs $k$.}
		\label{MinCost_Entropy}
		\end{figure*}

\section{Performance Evaluation}
\label{performance}
Our goal in this section is to answer the following question: What do the analytical expressions say in terms of how to distribute computation given the assumptions in Sects.\! \ref{compresstocompute}, \ref{networkmodel}? 

We first numerically compute the critical load thresholds $\rho_{th}$ in Prop. \ref{LoadThreshold} using (\ref{relaxed_condition_computation}), for special functions using the time complexity $d_f(M_v^c)$ given in (\ref{time_complexity_special_functions}) for different tasks, and illustrate them in Figure \ref{thresholdisolated}. The threshold $\rho_{th}$ increases with $M$, and gives a range $\rho<\rho_{th}$ where computation is allowed, i.e. $\gamma_f(\lambda_v^c)<\lambda_v^c$. The threshold is higher for high complexity functions such as ``classification", and lower for low complexity functions such as ``search". This is because given $\rho<\rho_{th}$, 
$d_f(M)$ for high complexity functions grows much faster than the communication cost of the generated flow. Low $\rho_{th}$ implies that node is computing even for small flow rates (or $M$) whereas high $\rho_{th}$ means that node can only compute if the flow is sufficient. Our numerical experiments show that the threshold $\rho_{th}$ is higher and converges faster for functions with high $d_f(M)$. However, if $d_f(M)$ is high, then a node can only compute when $M$ is sufficiently small such that a valid threshold exists. Benefit of computation increases from classification (fast convergence of $\rho_{th}$ to $1$), to MapReduce (modest convergence), to search (slow convergence). This implies that the limited computation resources should be carefully allocated to the computation of different tasks, with most of the resources being allocated to simple tasks. This result is valid when each node is considered in isolation, and enables handling of computation in a distributed manner. We next consider the network setting with mixing of flows. 
 
We numerically solve the MinCost problem in (\ref{costoptimization_expanded}) for some special functions. We assume randomized routing such that $P^{rou}(c)$ is a valid stochastic matrix, and that ${\bm \beta}^c$ and ${\bm \mu}^c$ are known. Furthermore, the values of $\Gamma_c$ and ${\bm \lambda}^c$ are also known apriori. In Figures \ref{Concave_MinCost_k}, \ref{Linear_MinCost_k}, and \ref{Convex_MinCost_k}, we investigate the behavior of MinCost versus computation cost scaling factor $k$, for search (or concave), MapReduce (or linear) \cite{LiAliYuAves2018}, and classification (or convex) functions, respectively (figures on R). We also investigate the trend of $\gamma_f(\lambda_v^c)$ as function of $k$ (figures on L). For all cases, it is intuitive that the processing factor $\gamma_f$ and MinCost should increase with $k$, and the rate of increase is determined by $\lambda$ and the time complexity of computation $d_f(M)$. We can also observe that as the entropic surjectivity $\Gamma_c$ of the function increases from $0.1$ to $0.8$, the function is less predictable, and the MinCost becomes higher as the nodes need to generate a higher processing factor $\gamma_f$. 


We next investigate the behavior of $\gamma_f(\lambda_v^c)$ versus $\Gamma_c$ (Figure \ref{MinCost_Entropy} L), and MinCost versus $\Gamma_c$ (Figure \ref{MinCost_Entropy} R) as a measure of surjectivity, where the values of $\Gamma_c$ and ${\bm \lambda}^c$ are coupled and have to be jointly determined. We observe that $\gamma_f(\lambda_v^c)$ is sensitive to surjectivity but not really to the time complexity (i.e. concave, convex, linear, etc) of the computation cost function. However, it is sensitive to the scaling $k$ as can be seen from Figure \ref{MinCost_Entropy} L. For $k=2$, $\gamma_f=\gamma_{f,\, LB}$ as computation is cheap, and for $k=10$, $\gamma_f>\gamma_{f,\, LB}$ as computation is expensive and the nodes need a higher processing factor. We expect that $\gamma_f(\lambda_v^c)$ increases in $\Gamma_c$ because it becomes harder to compress as the entropy of the function is higher and $C_{v,comp}^c$ should be higher.  
From Figure \ref{MinCost_Entropy} R, 
the behavior of MinCost is modified by the type and scaling of the computation cost function as well as $\Gamma_c$.  
Note that communication of the sources themselves gives an upper bound to the MinCost problem. Therefore, we normalize the MinCost with respect to the cost of communications only. When the computation is cheap, (e.g. concave, or linear with $k=2$ as in LHS of R), the nodes can compress the sources as long as $\Gamma_c$ is not high enough such that normalized MinCost is $1$. However, as the computation becomes costly (e.g. when $k=10$ as in RHS of R), compression of the sources does not minimize the overall cost. We also observe in RHS of this figure, for $k=10$, concave and linear cost functions are the same because the load of a node is smaller. 
The sources cannot be compressed beyond this because optimizing the total (communication and computation) cost is crucial. As entropic surjectivity goes up ($\Gamma_c\to 1$), implying that the function is surjective, we can infer that computing will not be allowed beyond a range where the value of $\Gamma_c$ is larger than a threshold. This is because allocating resources to computation does not incur less cost than communicating the entire source data. 

A node can perform computation and forward the processed data if the range of $\Gamma_c$ allowing compression is flexible. This is possible when computation is cheap. However, if a node's compression range is small, then the node simply relays most of the time. This indeed is the case when computation is very expensive. While computing at the source and communicating the end computation result might be feasible for some classes of functions, it might be very costly for some sets of functions due to the lack of cooperation among multiple sources. By making use of redundancy of data across geographically dispersed sources and the function to be computed, it is possible to decide how to distribute the computation in the network.

Our approach can be considered as a preliminary step for a better understanding of how to distribute computation in networks. Directions include devising coding techniques for in network functional compression, by blending techniques from compressed sensing to the Slepian and Wolf compression, and employing the concepts of graph entropy, and exploiting function surjectivity. They also include the extension to multi-class models with product-form distributions, allowing conversion among classes of packets when routed from/to a node.

\begin{appendix}

	The upper bound in (\ref{FlowBounds}) follows from the case of no computation. In this case, the long-term average number of packets in $v$ satisfies that $L_v^c=M_v^c=\frac{\lambda_v^c}{\mu_v^c(1-\rho_v^c)}$. However, when we allow function computation we expect to have $L_v^c\neq M_v^c$. 
	
	Assume that $\frac{1}{\lambda_v^c}d_f(M_v^c)\geq \frac{1}{\mu_v^c}$. If this assumption did not hold, we would have $d_f(M_v^c)
	<\frac{\lambda_v^c}{\mu_v^c-\gamma_f(\lambda_v^c)}$. For stability, the long-term average number of packets in $v$ waiting for communications service, i.e. $M_v^c$, should be upper bounded by the long-term average number of packets in $v$ waiting for computation service, i.e. $L_v^c-M_v^c$. Otherwise, $M_v^c$ will increase over time, which will violate the  stationarity assumption. 
	
	The lower bound follows from the definition of Little's law:
	\begin{align}
	L_v^c=\gamma_f(\lambda_v^c)\Big[\frac{1}{\lambda_v^c}d_f(M_v^c)+\frac{1}{\mu_v^c-\gamma_f(\lambda_v^c)}\Big]
	\overset{(a)}{\geq} \Hg(f_c(X_1^N)),\nonumber
	\end{align}

	\noindent where 
	$(a)$ is for recovering the function at the destination. Manipulating the lower bound relation above, we obtain 
	\begin{align}
	\gamma_f(\lambda_v^c)\geq 
	\mu_v^c\Big[\frac{\Hg(f_c(X_1^N))}{2}+1-\sqrt{\frac{\Hg(f_c(X_1^N))^2}{4}+1}\Big],\nonumber
	\end{align}
	using which we get the desired lower bound.

\end{appendix}
\balance

\bibliographystyle{IEEEtran}
\bibliography{Derya}

\end{document}